\documentclass[10pt,journal,compsoc]{IEEEtran}

\ifCLASSOPTIONcompsoc
  \usepackage[nocompress]{cite}
\else
  \usepackage{cite}
\fi

\ifCLASSINFOpdf
\else
\fi

\ifCLASSOPTIONcompsoc

\else

\fi

\ifCLASSINFOpdf

\else

\fi
\usepackage{float}
\usepackage[center]{caption2}
\usepackage{xspace}
\usepackage{balance}
\usepackage{algorithm}
\usepackage{algorithmic}
\usepackage{diagbox}
\usepackage{bbding}
\usepackage{amsthm}
\usepackage{amsmath,amssymb,amsfonts}
\usepackage{makecell}
\usepackage{multirow}
\usepackage{threeparttable}
\usepackage{subfigure}
\usepackage{booktabs}
\usepackage{tikz}
\newcommand{\Name}{\textbf{Hercules}\xspace}
\renewcommand{\algorithmicrequire}{\textbf{Input:}}
\renewcommand{\algorithmicensure}{\textbf{Output:}}
\newenvironment{packeditemize}{
\begin{list}{$\bullet$}{
\setlength{\labelwidth}{8pt}
\setlength{\itemsep}{0pt}
\setlength{\leftmargin}{\labelwidth}
\addtolength{\leftmargin}{\labelsep}
\setlength{\parindent}{0pt}
\setlength{\listparindent}{\parindent}
\setlength{\parsep}{0pt}
\setlength{\topsep}{3pt}}}{\end{list}}
\newtheorem{theorem}{Theorem}[section]

\newtheorem{lemma}[theorem]{Lemma}
\newcommand*{\circled}[1]{\lower.7ex\hbox{\tikz\draw (0pt, 0pt)
    circle (.4em) node {\makebox[0.9em][c]{\small #1}};}}

\renewcommand{\raggedright}{\leftskip=0pt \rightskip=0pt plus 0cm}
\raggedright
\begin{document}
\title{Hercules: Boosting the Performance of Privacy-preserving Federated Learning}

\author{{Guowen~Xu, Xingshuo~Han,  Shengmin~Xu, Tianwei~Zhang, Hongwei~Li, Xinyi~Huang, Robert~H.~Deng,~\IEEEmembership{Fellow,~IEEE}}

\IEEEcompsocitemizethanks{\IEEEcompsocthanksitem Guowen~Xu, Xingshuo~Han, and Tianwei~Zhang are with the School of Computer Science and Engineering, Nanyang Technological University. (e-mail: guowen.xu@ntu.edu.sg; xingshuo001@e.ntu.edu.sg; tianwei.zhang@ntu.edu.sg)

\IEEEcompsocthanksitem Shengmin~Xu  and Xinyi~Huang are with the College of Computer and Cyber Security, Fujian
Normal University, Fuzhou, China (e-mail: smxu1989@gmail.com; xyhuang81@gmail.com)

\IEEEcompsocthanksitem Hongwei~Li  is with the school of Computer Science and Engineering,  University of Electronic Science and Technology of China, Chengdu 611731, China.(e-mail: hongweili@uestc.edu.cn)

\IEEEcompsocthanksitem  Robert~H.~Deng  is with the School of Information Systems, Singapore Management University, 178902 Singapore
(e-mail:robertdeng@smu.edu.sg)}}

 \IEEEcompsoctitleabstractindextext{
\begin{abstract}
 \raggedright
 In this paper, we address the problem of privacy-preserving  federated neural network training with  $N$ users. We present \Name, an  efficient and high-precision  training framework that can tolerate collusion of up to $N-1$ users.  \Name follows the POSEIDON framework proposed by Sav  \textit{et al.} (NDSS'21), but makes a qualitative leap in performance with the following contributions: (i) we design  a novel parallel homomorphic computation method for matrix operations, which  enables  fast Single Instruction and Multiple Data (SIMD) operations  over ciphertexts.  For the multiplication of two $h\times h$ dimensional matrices, our method reduces the computation complexity from $O(h^3)$ to $O(h)$. This greatly improves the training efficiency of the neural network since the ciphertext computation is dominated by the convolution operations; (ii) we present an efficient approximation on the sign function  based on the composite polynomial approximation. It is used to approximate non-polynomial functions (i.e.,  \texttt{ReLU} and \texttt{max}),  with  the optimal asymptotic complexity.  Extensive experiments on various benchmark datasets (BCW, ESR, CREDIT, MNIST, SVHN, CIFAR-10 and CIFAR-100) show that compared with  POSEIDON, \Name obtains up to $4\%$ increase in model accuracy, and up to $60\times$ reduction in the computation and communication cost.
\end{abstract}
}
\maketitle

\IEEEdisplaynotcompsoctitleabstractindextext

\IEEEpeerreviewmaketitle

\section{Introduction}
As a promising neural network training mechanism, Federated Learning (FL) has been highly sought after with some attractive features including amortized overhead and mitigation of privacy threats. However, the conventional FL setup has some inherent privacy issues \cite{mouchet2020multiparty, chen2019efficient}.
Consider  a scenario where a company (referred to as the cloud server) pays multiple users and requires them to train a target neural network model collaboratively.  Although each user is only required to upload the intermediate data (e.g., gradients) instead of the original training data to the server during the training process, a large amount of sensitive information can still be leaked implicitly from these intermediate values.  Previous works have demonstrated many powerful attacks to achieve this, such as attribute inference attacks and gradient reconstruction attacks \cite{geiping2020inverting,zhu2020deep,chen2020gan}.
On the other hand, the target model is locally distributed to each user according to the FL protocol, which ignores the  model privacy and may be impractical in  real-world scenarios. Actually, to protect the model privacy, the server must keep users ignorant of the details of the model parameters throughout the training process.

\subsection{Related Works}
\label{Related Works and  Challenges}
Extensive works have been proposed to mitigate the above privacy threats. In general,  existing privacy-preserving deep learning solutions mainly rely on the following two lines of technologies: \textit{Differential Privacy} (DP) \cite{abadi2016deep,yu2019differentially} and  \textit{crypto-based multiparty secure computing} (MPC) \cite{mohassel2018aby3,agrawal2019quotient,rathee2020cryptflow2,zheng2019helen,sav2020poseidon}. Each one has merits and demerits depending on the scenario to which it is applied.

\noindent\textbf{Differential Privacy.}
DP is usually applied in the training phase \cite{abadi2016deep,yu2019differentially}. To ensure the indistinguishability between individual samples while maintaining high training accuracy, each user is required to add noise to the gradient or local parameters that meets the preset privacy budget.  Abadi \textit{et al.} \cite{abadi2016deep} propose the first differentially private stochastic gradient descent (SGD) algorithm. They carefully implement gradient clipping,  hyperparameter tuning, and moment accountant to obtain a tight estimate of overall privacy loss, both asymptotically and empirically. Yu \textit{et al.} \cite{yu2019differentially} design a new DP-SGD, which employs a new primitive called  zero concentrated differential privacy (zCDP) for privacy accounting, to achieve a rigorous estimation of the privacy loss. In recent years, many variants of the above works have been designed and applied to specific scenarios \cite{wei2020federated,wu2020value,bagdasaryan2019differential,wang2019sparse}. Most of them follow the principle that the minimum accumulated noise is added to the gradient or local parameters while meeting the preset privacy budget.

DP is cost-effective because each user is only required to add noise that obeys a specific distribution during training. However, it is forced to make a trade-off between training accuracy and privacy, i.e., a strong privacy protection level can be reached at the cost of certain model accuracy drop \cite{jayaraman2019evaluating,hitaj2017deep}. This goes against the motivation of this paper, as our goal is to design a highly secure FL training framework without compromising the model accuracy.


\noindent\textbf{Crypto-based multiparty secure computing.} The implementation of this strategy mainly relies on two general techniques, secret sharing \cite{patra2021aby2} and homomorphic encryption (HE) \cite{zheng2019helen}. MPC enables the calculation of arbitrary functions collaboratively by multiple parties without revealing the secret input of each party. To support  privacy-preserving neural network training,  most existing works \cite{patra2021aby2,mohassel2017secureml,mohassel2018aby3,agrawal2019quotient,rathee2020cryptflow2} rely on splitting the training task into two or more servers, who are usually assumed to be non-colluding.  Then, state-of-the-art secret sharing methods, including arithmetic sharing \cite{patra2021aby2}, boolean sharing \cite{mohassel2018aby3}, and Yao's garbled circuit \cite{kumar2020cryptflow} are carefully integrated to efficiently implement various mathematical operations under the ciphertext. Mohassel \textit{et al.} \cite{mohassel2017secureml} propose SecureML, the first privacy-preserving machine learning framework for generalized linear model regression and neural network training.  It lands on the setting of two non-colluding servers, where users securely outsource local data to them. Then, several types of secret sharing methods are mixed and used to complete complex ciphertext operations. 
Other works, \textit{e.g.}, ABY$^3$ \cite{mohassel2018aby3}, QUOTIENT\cite{agrawal2019quotient}, BLAZE \cite{patrablaze}, Trident\cite{Trident}, are also exclusively based on the MPC protocol between multiple non-colluding servers (or a minority of malicious servers) to achieve fast model training and prediction.

It is  cost-effective to outsource the training task among multiple users to several non-colluding servers, avoiding the high communication overhead across large-scale users. However, it may be impractical in real scenarios where the setting of multiple servers is not available. Especially in FL scenarios, users are more inclined to keep their datasets locally rather than uploading data to untrusted servers. To alleviate this problem, several works \cite{zheng2019helen,sav2020poseidon,chen2019efficient,froelicher2020scalable} propose to use multi-party homomorphic encryption (a.k.a. threshold homomorphic encryption, as a variant of the standard HE), as the underlying technology to support direct interactions among multiple data owners for distributed learning. For example, Zheng \textit{et al.} \cite{zheng2019helen} present Helen, a secure distributed  learning approach for linear models,  where the threshold Paillier scheme \cite{galbraith2002elliptic} is used to protect users' local data.  Froelicher \textit{et al.} \cite{froelicher2020scalable} reduce the computation overhead of Helen by using the packed plaintext encoding with the SIMD technology \cite{chen2019efficient}. Sav \textit{et al.} propose {POSEIDON} \cite{sav2020poseidon}, the first distributed training framework with multi-party homomorphic encryption.  It relies on the multiparty version of the CKKS (MCKKS) cryptosystem \cite{cheon2017homomorphic} to encrypt users' local data. Compared with  the standard CKKS, the secret key  of MCKKS is securely shared with multiple entities.  As a result, each entity still performs the function evaluation under the same public key. However, the decryption of the result requires the participation of all entities. Besides, non-polynomial functions are approximated as polynomial functions  to be efficiently  executed by CKKS.

\subsection{Technical Challenges}

 In this paper, we follow the specifications of POSEIDON to design our FL training framework, because such a technical architecture enables the users' data to be kept locally without incurring additional servers.
 However, there are still several critical issues that have not been solved well. (1) Computation overhead is the main obstacle hindering the development of HE. It usually requires more computing resources to perform the same machine learning tasks compared to outsourcing-based solutions \cite{mohassel2018aby3,agrawal2019quotient,rathee2020cryptflow2}. Although there are some optimization methods such as parameter quantization and model compression \cite{agrawal2019quotient,zhang2020batchcrypt}, they inevitably degrade the model accuracy. Recently, Zhang \textit{et al.}  \cite{zhang2021gala} design GALA, which employs a novel coding technique for matrix-vector multiplication.  In this way, multiple plaintexts are packed into one ciphertext to perform efficient homomorphic SIMD operations without reducing the calculation accuracy. However, GALA is specifically designed for the MPC protocol that uses a mixture of HE and garbled circuits, and its effectiveness is highly dependent on the assistance of the inherent secret sharing strategy. Therefore, it is necessary to design a computation optimization method that is completely suitable for HE, without sacrificing the calculation accuracy.
 (2) There is a lack of satisfactory approximation mechanisms for non-polynomial functions in HE. HE basically supports homomorphic addition and multiplication. For non-polynomial functions, especially \texttt{ReLU}, one of the most popular activation functions in  hidden layers,  we need to approximate them to polynomials for ciphertext evaluation. The common polynomial approximation method, such as the minimax method, aims to find the approximate polynomial with the smallest degree on the objective function under the condition of a given error bound. However, the computation complexity of evaluating these polynomials is enormous, making it quite inefficient to obtain the fitting function with high-precision \cite{iliashenko2021faster,cheon2020efficient}. Recently, Lu \textit{et al.} \cite{lu2020pegasus} propose PEGASUS, which can efficiently switch back and forth between a packed CKKS ciphertext and FHEW ciphertext \cite{ducas2015fhew} without decryption, allowing us to evaluate both polynomial and non-polynomial functions on encrypted data.  However, its performance is still far from practical.

\subsection{Our Contributions}
\label{Our Contributions}
As discussed above, the HE-based FL is more in line with the needs of most real-world applications, compared to other methods. However, it suffers from computing bottlenecks and poor compatibility with non-polynomial functions. To mitigate these limitations,  we present \Name, an  efficient, privacy-preserving and high-precision framework for FL. \Name  follows the tone of the state-of-the-art work POSEIDON \cite{sav2020poseidon}, but  makes a qualitative leap in performance.  Specifically, we first devise a new method for parallel homomorphic computation of matrix, which supports fast homomorphic SIMD operations, including addition, multiplication, and transposition.  Then, instead of fitting the replacement function of \texttt{ReLU} for training in POSEIDON, we design an efficient  method based on the composite polynomial approximation. In short, the contributions of \Name are summarized as follows:

\begin{packeditemize}

\item  We design a new method to execute matrix operations in parallel, which can pack multiple  plaintexts into a ciphertext to achieve fast homomorphic SIMD operations (Section~\ref{Parallelized Matrix Homomorphic Operations}).  Our key insight is to minimize the number of plaintext slots that need to be rotated in matrix multiplication through customized permutations. Compared with existing works \cite{sav2020poseidon, halevi2015bootstrapping},  our solution  reduces the computation complexity from $O(h^3)$ to $O(h)$ for the multiplication of any two $h\times h$ matrices.  It greatly improves the neural network training efficiency since the ciphertext computation  is dominated by the convolution operations. We describe the detail of efficiently executing matrix transposition on packed ciphertexts, and packing multiple matrices into one ciphertext, yielding better-amortized performance.

\item We present  an efficient approximation on the sign function  based on the composite polynomial approximation, with optimal asymptotic complexity (Section \ref{Approximation for Comparison Function}). The  core of our solution is to carefully construct a polynomial $g$ with a constant degree, and then make the composite polynomial $g\circ g\circ g\circ \cdots \circ g$ infinitely close to the sign function, as the number of $g$ increases.  In this way, our new algorithm only requires $\Theta(\log (1/\delta))+\Theta(\log \sigma)$ computation complexity to obtain an approximate sign function result of  $m\in[-1,-\delta]\cup [\delta, 1]$  within $2^{-\sigma}$ error. For example, for an encrypted 20-bit integer $m$, we can obtain the result of the sign function within $2^{-20}$ error with an amortized running time of 20.05 milliseconds, which is $33\times$  faster than the state-of-the-art work \cite{cheon2019numerical}.

\item We show that \Name provides semantic security in the FL scenario consisting of $N$ users and a parameter server, and tolerates collusion among up to $N-1$ passive users (Section \ref{Implementation of Hercules}). This is mainly inherited from the property of the MCKKS.

\item We conduct extensive experiments on various benchmark datasets (BCW, ESR, CREDIT, MNIST, SVHN, CIFAR-10 and CIFAR-100) to demonstrate the superiority of \Name in terms of classification accuracy, and overhead of computation and communication (Section \ref{sec:PERFORMANCE EVALUATION}). Specifically, compared with  POSEIDON, we obtain up to $4\%$ increase in model accuracy, and up to $60\times$ reduction in the computation and communication cost. 

\end{packeditemize}

\noindent\textbf{Roadmap}:  In  Section \ref{sec:PROBLEM STATEMENT}, we review some basic concepts used in this paper,  and introduce the scenarios and threat models. In  Sections \ref{Parallelized Matrix Homomorphic Operations} to \ref{Implementation of Hercules}, we give the details of \Name.   Performance evaluation is presented in  \ref{sec:PERFORMANCE EVALUATION}.  Section \ref{sec:conclusion} concludes the paper.
\section{Preliminaries}
\label{sec:PROBLEM STATEMENT}

\subsection{Neural Network Training}
\label{sec:Neural Network Training}
A neural network usually consists of an input layer, one or more hidden layers, and an output layer, where hidden layers  include convolutional layers, pooling layers, activation function layers, and fully connected layers. The connections between neurons in adjacent layers are parameterized by ${\omega}$ (\textit{i.e.}, model parameters), and each neuron is associated with an element-wise activation function $\varphi$ (such as sigmoid, \texttt{ReLU}, and softmax).
Given the training sample set $({x}, {y})\in D$,  training  a neural network of $\mathbb{L}$ layers is generally divided into two phases: \textit{feedforward} and \textit{backpropagation}. Specifically, at the $k$-th iteration, the weights between layers $j$ and $j+1$ are denoted as a matrix ${\omega}_{j}^{k}$;  matrix $M_j$ represents the activation of neurons in the $j$-th layer. Then the input ${x}$ is sequentially propagated  to each layer with operations of linear transformation (i.e, $E_j^{k}={\omega}_{j}^{k}\times M_{j-1}^{k}$) and non-linear transformation (i.e., $M_j^{k}=\varphi(E_j^{k})$) to obtain the final classification result ${\bar{y}}=M_\mathbb{L}^{k}$. With the loss function ${L}$ which is usually set as ${L}$=$||{y}-\bar{{y}}||_2$,   the mini-batch based Stochastic Gradient Descent (SGD) algorithm  \cite{sav2020poseidon} is  exploited to optimize the parameter ${\omega}$. The parameter update rule is $ {\omega}_{j}^{k+1}={\omega}_{j}^{k}-\frac{\eta}{\mathcal{B}}\bigtriangledown {\omega}_{j}^{k}$, where $\eta$ and $\mathcal{B}$ indicate the learning rate and the random batch size of input samples, and $\bigtriangledown {\omega}_{j}^{k}=\frac{\partial L}{\partial {\omega}_{j}^{k}}$. Since the transposition of matrices/vectors is  involved in the \textit{backpropagation}, we use $V^{T}$ to represent the transposition of variable $V$. The \textit{feedforward} and \textit{backpropagation} steps are performed iteratively until the neural network meets the given convergence constraint. The detailed implementation is shown in \textbf{Algorithm}~\ref{algorithm 1}.
\renewcommand{\thempfootnote}{\alph{mpfootnote}}
 \begin{algorithm}
\small
\caption{ Mini-batch based SGD algorithm }
\label{algorithm 1}
\begin{algorithmic}[1]
\REQUIRE  ${\omega}_{1}^{k}, {\omega}_{2}^{k}, \cdots, {\omega}_{\mathbb{L}}^{k}$.
\ENSURE  ${\omega}_{1}^{k+1}, {\omega}_{2}^{k+1}, \cdots, {\omega}_{\mathbb{L}}^{k+1}$.
\FOR{{$t=1$ to $\mathcal{B}$}}
\STATE $M_0=X[t]$ \qquad \qquad \qquad \qquad \textit{$\rhd$ feedforward}
\FOR{{$j=1$ to $\mathbb{L}$}}
\STATE $E_j^{k}={\omega}_{j}^{k}\times M_{j-1}^{k}$
\STATE $M_j^{k}=\varphi(E_j^{k})$
\ENDFOR
\STATE $L_\mathbb{L}^{k}=||y[t]-M_\mathbb{L}^{k}||_{2}$\qquad \qquad \textit{$\rhd$ backpropagation}
\STATE $L_\mathbb{L}^{k}=\varphi'(E_\mathbb{L}^{k})\odot L_\mathbb{L}^{k}$ \footnotemark{}
\STATE $\bigtriangledown {\omega}_{\mathbb{L}}^{k}+=(M_{\mathbb{L}-1}^{k})^{T}\times L_\mathbb{L}^{k}$
\FOR{{$j=\mathbb{L}-1$ to $1$}}
\STATE $L_j^{k}=L_{j+1}^{k}\times (\omega_{j+1}^{k})^{T}$
\STATE $L_j^{k}=\varphi'(E_j^{k})\odot L_j^{k}$
\STATE $\bigtriangledown {\omega}_{j}^{k}+=(M_{j-1}^{k})^{T}\times L_j^{k}$
\ENDFOR
\ENDFOR
\FOR{{$j=1$ to $\mathbb{L}$}}
\STATE ${\omega}_{j}^{k+1}={\omega}_{j}^{k}-\frac{\eta}{\mathcal{B}}\bigtriangledown {\omega}_{j}^{k}$
\ENDFOR
\end{algorithmic}
\end{algorithm}
\footnotetext{$\varphi'(\cdot)$ and $\odot$ indicate partial derivative and element-wise product.}

\subsection{Multiparty Version of CKKS }
\label{Multiparty Version of CKKS Homomorphic Encryption}
\Name relies on the multiparty version of Cheon-Kim-Kim-Song (MCKKS) \cite{sav2020poseidon} fully homomorphic encryption to protect users' data as well as the model's parameter privacy. Compared with  the standard CKKS, the secret key  of MCKKS is securely shared with all entities.  As a result, each entity still performs ciphertext evaluation under the same public key, while the decryption of the result requires the participation of all entities. As shown in \cite{sav2020poseidon}, MCKKS has several attractive properties: (i) it is naturally suitable for floating-point arithmetic circuits, which facilitates the implementation of machine learning; (ii) it flexibly supports collaborative computing among multiple users without revealing the respective share of the secret key; (iii) it supports the function of key-switch, making it possible to convert a ciphertext encrypted under a public key into a ciphertext under another public key without decryption. Such a property facilitates the decryption of ciphertexts collaboratively.  We provide a short description of MCKKS and list all the functions required by \Name in Figure~\ref{MCKKS}.  Informally, given a cyclotomic polynomial ring with a dimension of $\mathcal{N}$, the plaintext and ciphertext space of MCKKS is defined as $R_{Q_
\mathcal{L}}=\mathbb{Z}_{Q_\mathcal{L}}[X]/(X^{\mathcal{N}}+1)$, where $Q_\mathcal{L}=\prod_{0}^{\mathcal{L}}q_i$, and each $q_i$ is a unique prime. $Q_\mathcal{L}$ is the ciphertext module under the initial level $\mathcal{L}$. In CKKS, a plaintext vector with up to $\mathcal{N}/2$ values can be encoded into a ciphertext.  As shown in Figure~\ref{MCKKS}, given a plaintext $m\in R_{Q_\mathcal{L}}$ (or a plaintext vector $\mathbf{m}=(m_1, \cdots, m_n)\in R_{Q_\mathcal{L}}^{n}$, with $n\leq \mathcal{N}/2$) with its encoded (packed) plaintext $\hat{m}$, the corresponding  ciphertext is denoted as $[\mathbf{c}]_{pk}=(c_1, c_2)\in R_{Q_\mathcal{L}}^{2}$. Besides, we use  symbols $\mathcal{L}_{\mathbf{c}_{pk}}$, $\Delta_{\mathbf{c}_{pk}}$, $\mathcal{L}$, $\Delta$ to indicate the current level of $[\mathbf{c}]_{pk}$, the current scale of $\mathbf{c}$, the initial level, and the initial scale of a fresh ciphertext, respectively.  All functions named starting with $\mathbf{D}$ (except for $\mathbf{Dcd}(\cdot)$) in Figure~\ref{MCKKS} need to be executed cooperatively by all the users, while the rest operations can be executed locally by each user with the public key. For more details about MCKKS,  please refer to literature \cite{sav2020poseidon,froelicher2020scalable,mouchet2020multiparty}.
\renewcommand\tablename{Fig.}
\renewcommand \thetable{\arabic{table}}
\setcounter{table}{0}
\begin{table}[!htb]
\small
\centering
\begin{tabular}{|p{8.2cm}|}
\hline
\begin{itemize}

 \item [1)]$\mathbf{SecKeyGen}(1^{\lambda})$: Given a security parameter $\lambda$,  output a secret key $sk_i$ for each user $i\in[N]$, where $[N]$ is the shorthand $\{1, 2, \cdots N\}$ and  $\sum_{i=1}^{i=N}sk_i=sk$.
\item [2)] $\mathbf{DKeyGen}(\{sk_i\})$: Given the set of secret keys $\{sk_i\}$, $i\in[N]$, output the collective public key $pk$.
 \item [3)]$\mathbf{Ecd}(\cdot)$: Given a plaintext  $m$ (or a plaintext vector $\mathbf{m}$ whose dimension does not exceed $\mathcal{N}/2$), output the encoded (packed) plaintext $\hat{m}\in R_{Q_\mathcal{L}}$, with scale $\Delta$.
  \item [4)]$\mathbf{Dcd}(\hat{m})$: Given an encoded (packed) plaintext $\hat{m}\in R_{Q_{\mathcal{L}_m}}$ with scale $\Delta_m$, output the decoding of $m$ (or the plaintext vector $\mathbf{m}$).
  \item [5)]$\mathbf{Enc}(pk,\hat{m})$: Given the collective public key $pk$, and an encoded (packed) plaintext $\hat{m}\in R_{Q_{\mathcal{L}}}$, output the ciphertext $\mathbf{[c]}_{pk}\in R_{Q_\mathcal{L}}^{2}$ with scale $\Delta$.
  \item [6)]$\mathbf{DDec}(\mathbf{[c]}_{pk}, \{sk_i\})$: Given a  ciphertext $\mathbf{[c]}_{pk}\in R_{Q_{\mathcal{L}_\mathbf{c}}}^2$ with scale $\Delta_{\mathbf{c}_{pk}}$, and the set of secret keys $\{sk_i\}$, $i\in[1, N]$, output the plaintext $p\in R_{Q_{\mathcal{L}_\mathbf{c}}}$ with scale $\Delta_{\mathbf{c}_{pk}}$.
  \item [7)]$\mathbf{Add}(\mathbf{[c]}_{pk}, \mathbf{[c']}_{pk})$: Given two ciphertexts $\mathbf{[c]}_{pk}$ and $\mathbf{[c']}_{pk}$ encrypted with the same public key $pk$, output $[\mathbf{c}+\mathbf{c'}]_{pk}$ with level $\min(\mathcal{L}_{\mathbf{c}_{pk}}, \mathcal{L}_{\mathbf{c}'_{pk}})$ and scale $\max(\Delta_{\mathbf{c}_{pk}}, \Delta_{\mathbf{c}'_{pk}})$.
  \item [8)]$\mathbf{Sub}(\mathbf{[c]}_{pk}, \mathbf{[c']}_{pk})$: Given two ciphertexts $\mathbf{[c]}_{pk}$ and $\mathbf{[c']}_{pk}$, output $[\mathbf{c}-\mathbf{c}']_{pk}$ with level $\min(\mathcal{L}_{\mathbf{c}_{pk}}, \mathcal{L}_{\mathbf{c}'_{pk}})$ and scale $\max(\Delta_{\mathbf{c}_{pk}}, \Delta_{\mathbf{c}'_{pk}})$.
  \item [9)]$\mathbf{Mul}_{pt}([\mathbf{c}]_{pk}, \hat{m})$: Given a ciphertext $[\mathbf{c}]_{pk}$  and an encoded (packed) plaintext $\hat{m}$, output $[\mathbf{c}m]_{pk}$ with level $\min(\mathcal{L}_{\mathbf{c}_{pk}}, \mathcal{L}_{\mathbf{c}'_{pk}})$ and scale $\Delta_{\mathbf{c}_{pk}} \times \Delta_m$.
  \item [10)]$\mathbf{Mul}_{ct}([\mathbf{c}]_{pk}, [\mathbf{c}']_{pk})$: Given two ciphertexts $[\mathbf{c}]_{pk}$ and $[\mathbf{c}']_{pk}$, output $[\mathbf{c}\mathbf{c}']_{pk}$ with level $\min(\mathcal{L}_{\mathbf{c}_{pk}}, \mathcal{L}_{\mathbf{c}'_{pk}})$ and scale $\Delta_{\mathbf{c}_{pk}} \times \Delta_{\mathbf{c}'_{pk}}$.
 \item [11)]$\mathbf{Rot}([\mathbf{c}]_{pk}, k)$: Given a ciphertexts $[\mathbf{c}]_{pk}$, homomorphically rotate $[\mathbf{c}]_{pk}$  to the right ($k>0$) or to the left ($k<0$) by $k$ times.
  \item [12)]$\mathbf{RS}([\mathbf{c}]_{pk})$: Given a ciphertexts $[\mathbf{c}]_{pk}$, output $[\mathbf{c}]_{pk}$ with scale $\Delta_\mathbf{c}/q_{\Delta_\mathbf{c}}$ and level $\mathcal{L}_\mathbf{c}-1$.
 \item [13)]$\mathbf{DKeySwitch}([\mathbf{c}]_{pk}, pk', \{sk_i\})$: Given a ciphertexts $[\mathbf{c}]_{pk}$, another public key $pk'$, and the set of secret keys $\{sk_i\}$, $i\in[N]$,  output $[\mathbf{c}]_{pk'}$.
  \item [14)]$\mathbf{DBootstrap}([\mathbf{c}]_{pk}, \mathcal{L}_{\mathbf{c}_{pk}}, \Delta_{\mathbf{c}_{pk}}, \{sk_i\})$: Given a ciphertexts $[\mathbf{c}]_{pk}$ with level $\mathcal{L}_{\mathbf{c}_{pk}}$ and scale $\Delta_{\mathbf{c}_{pk}}$, and the set of secret keys $\{sk_i\}$, $i\in[N]$,  output $[\mathbf{c}]_{pk}$ with initial $\mathcal{L}$ and scale $\Delta$.
 \end{itemize}\\
    \hline
\end{tabular}
\caption{Cryptographic operations of MCKKS}
\label{MCKKS}
\end{table}

\subsection{Threat Model and Privacy Requirements}
\label{Threat Model and Privacy Requirements}
 We consider a FL scenario composed of a parameter server and $N$ users for training a neural network model collaboratively.  Specifically, the server (also the model owner) first initializes the target model $\mathcal{M}$ and broadcasts the encrypted model $[\mathbf{M}]_{pk}=\mathbf{Enc}(pk,\mathcal{M})$ (i.e., encrypting all the model parameters) to all the users\footnote{Note that the server knows nothing about the secret key $sk$ corresponding to $pk$. $sk$ is securely shared with $N$ users and can only be restored with the participation of all the users.}. Then, each user  $P_i$ with a dataset $\{x, y\}\in D_i$  trains $[\mathbf{M}]_{pk}$ locally using the mini-batch SGD algorithm and then sends the encrypted local gradients to the server. After receiving the gradients from all the users, the server homomorphically aggregates them and broadcasts back the global model parameters. All the participants perform the above process iteratively until the model converges. Since the final trained model is encrypted with the public key $pk$, for the accessibility of the server to the plaintext model, we rely on the function $\mathbf{DKeySwitch}$ (Figure~\ref{MCKKS}), which enables  the conversion of  $[\mathbf{M}]_{pk}$ under the  public key $pk$ into $[\mathbf{M}]_{pk'}$ under the server's public key $pk'$ without decryption (refer to Section~\ref{Implementation of Hercules} for more details). As a result, the server obtains the plaintext model by decrypting $[\mathbf{M}]_{pk'}$ with its secret key.

In  \Name, we consider a passive-adversary model with collusion of up to $N-1$ users\footnote{See Appendix~\ref{Security Extensions} for more discussion about malicious adversary model.}. Concretely, the server and each user abide by the agreement and perform the training procedure honestly. However,   there are two ways of colluding in  \Name by sharing their own inputs, outputs and observations during the training process for different purposes: (i) collusion among up to $N-1$ users to derive the training data of other users or the model parameters of the server; (ii) collusion among the server and no more than $N-1$ users to infer the training data of other users. Given such a threat model, in the training phase, the privacy requirements  of \Name are defined as below:
\begin{packeditemize}
\item  \textbf{Data privacy}: No participant (including the server) should learn more information about the input data (e.g., local datasets, intermediate values, local gradients) of other honest users, except for the information that can be inferred from its own inputs and outputs.

\item \textbf{Model privacy}: No user  should learn more information about the parameters of the model, except for information that can be inferred from its own inputs and outputs.
\end{packeditemize}
In Section~\ref{Implementation of Hercules}, we will provide  (sketch) proofs of these privacy requirements with the real/ideal simulation formalism \cite{canetti2014practical}.

\section{Parallelized Matrix Homomorphic Operations}
\label{Parallelized Matrix Homomorphic Operations}
\Name essentially exploits MCCK as the underlying architecture to implement privacy-preserving federated neural network training. Since the vast majority of the computation of a neural network consists of convolutions (equivalent to matrix operation), \Name is required to handle this type of operation homomorphically very frequently. In this section, we describe our optimization method to perform homomorphic matrix operations in a parallelized manner, thereby substantially improving the computation  performance of HE.
\subsection{Overview}
 At a high level, operations between two matrices, including multiplication and transposition, can be decomposed into a series of combinations of linear transformations.  To handle homomorphic matrix operations in an SIMD manner, a straightforward way is to directly perform the relevant linear operations under the packed ciphertext (Section~\ref{Preliminary Knowledge}). However, it is computationally intensive and requires $O(h^3)$ computation complexity for the multiplication of two $h\times h$-dimensional matrices (Section~\ref{Permutation for matrix multiplication}). Existing state-of-the-art methods \cite{halevi2015bootstrapping} propose to transform the multiplication of two $h\times h$-dimensional matrices into  inner products between multiple vectors. It can reduce the complexity from $O(h^3)$ to $O(h^2)$, however,  yielding $h$ ciphertexts to represent a matrix (Section~\ref{Homomorphic matrix multiplication}). Compared to existing efforts, our method  only needs $O(h)$ complexity and derives one ciphertext. Our key insight is to first formalize the linear transformations corresponding to matrix operations, and then tweak them to minimize redundant operations in the execution process. In the following  we present the technical details of our method. To facilitate understanding, Figure~\ref{Examples} also provides an intuitive example, where the detailed steps of the multiplication of two $3\times 3$-dimensional matrices are described for comprehensibility.



\subsection{Preliminary Knowledge}
\label{Preliminary Knowledge}
 We first introduce some useful symbols and concepts. Specifically, all the vectors in this section refer to row vectors, and are represented in bold (e.g., $\mathbf{a}$).  As shown in Figure~\ref{MCKKS},  given a plaintext vector $\mathbf{m}=(m_1, \cdots, m_n)\in R_{Q_\mathcal{L}}^{n}$, with $n\leq \mathcal{N}/2$,  CKKS enables to encode  the plaintext vector $\mathbf{m}$ into an encoded plaintext $\hat{m}\in R_{Q_\mathcal{L}}$, where each $m_i$, $i\in[n]$ has a unique position called a plaintext slot in the encoded  $\hat{m}$. Then, $\hat{m}$ is encrypted as a ciphertext $[\mathbf{c}]_{pk}$. Hence, performing arithmetic operations (including addition and multiplication) on $[\mathbf{c}]_{pk}$ is equivalent to doing the same operation on every plaintext slot at once.

The ciphertext packing technology is capable of packing multiple plaintexts into one ciphertext and realizing the homomorphic SIMD operation, thereby effectively reducing the space and time complexity of encryption/calculation of a single ciphertext. However, it is incapable of handling the arithmetic circuits when some inputs are in different plaintext slots. To combat that, CKKS provides a rotation function $\mathbf{Rot}([\mathbf{c}]_{pk}, k)$. Given a ciphertext $[\mathbf{c}]_{pk}$  of  a plaintext vector $\mathbf{m}=(m_1, \cdots, m_n)\in R_{Q_\mathcal{L}}^{n}$,  $\mathbf{Rot}([\mathbf{c}]_{pk}, k)$ transforms $[\mathbf{c}]_{pk}$ into an encryption of $\mathbf{R}(\mathbf{m}, k):= (m_k, \cdots, m_{n-1},$$ m_{0},\cdots, m_{k-1})$.
 $k$ can be either positive or negative and we have a rotation by $\mathbf{R}(\mathbf{m}, k)= \mathbf{R}(\mathbf{m}, n-k)$. 

Based on the above explanation, we  adopt a method proposed by Shai \textit{et al.} \cite{halevi2014algorithms}, which supports arbitrary linear transformations for encrypted vectors. Specifically, an arbitrary linear transformation $\mathcal{T}: R^{n}\rightarrow R^{n}$ on the plaintext vector can be expressed as $\mathcal{T}: \mathbf{m}\rightarrow U\cdot\mathbf{m} $ using some matrix $U\in R^{n\times n}$. This process can be implemented in ciphertext by
the rotation function $\mathbf{Rot}$ and constant multiplication operation $\mathbf{Mul}_{pt}$. Concretely, for $0\leq k<n$, a $k$-th \textit{diagonal vector} $U$ is defined as $\mathbf{u}_k=(U_{0,k}, U_{1, k+1}, \cdots, U_{n-k-1, n-1}, U_{n-k,0},\cdots, U_{n-1, k-1})\in R^{n}$.  Consequently, we have
\begin{small}
 \begin{equation}
 \label{eq1}
\begin{split}
U\cdot \mathbf{m}=\sum_{0\leq k<n}\mathbf{u}_k\odot \mathbf{R}(\mathbf{m},k).
\end{split}
\end{equation}
\end{small}
Hence, given the matrix $U$, and a ciphertext $[\mathbf{c}]_{pk}$ of the vector $\mathbf{m}$, \textbf{Algorithm}~\ref{algorithm 2} shows the details of computing encrypted $U\cdot \mathbf{m}$.
We observe that \textbf{Algorithm}~\ref{algorithm 2} requires $n$ additions, constant multiplications and rotations. Because the  rotation operation is much more intensive than the other two  operations, the computation complexity of \textbf{Algorithm}~\ref{algorithm 2} is usually regarded as asymptotically $O(n)$ rotations.

\renewcommand{\algorithmicrequire}{\quad \; \textbf{procedure}}
\begin{algorithm}
\caption{ Homomorphic linear transformation }
\label{algorithm 2}
\begin{algorithmic}[1]
\small
\REQUIRE \texttt{HE-LinTrans} $([\mathbf{c}]_{pk}, U)$
\STATE $[\mathbf{c}']_{pk}\leftarrow \mathbf{Mul}_{pt}([\mathbf{c}]_{pk}, \mathbf{u}_0)$
\FOR{{$k=1$ to $n-1$}}
\STATE $[\mathbf{c}']_{pk}\leftarrow \mathbf{Add}([\mathbf{c}']_{pk}, \mathbf{Mul}_{pt}(\mathbf{Rot}([\mathbf{c}]_{pk}, k), \mathbf{u}_k))$
\ENDFOR
\STATE \textbf{return} $[\mathbf{c}']_{pk}$
\end{algorithmic}
\end{algorithm}

In the following, we first describe how to express the multiplication between two matrices by permutation. Then, we introduce an encoding method that converts a matrix into a vector. Based on this, we describe the details of matrix multiplication on packed ciphertexts.

\subsection{Permutation for Matrix Multiplication}
\label{Permutation for matrix multiplication}
 Given a $(h\times h)$-dimensional matrix $A=(A_{i,j})_{0\leq i,j<h}$, we describe  four permutation operations ($\mu$, $\zeta$, $\phi$, $\pi$) on it. For simplicity,  we use $\mathbb{Z}\cap[0,h)$ to denote the representative of $\mathbb{Z}_h$,  $[i]_h$ indicates the reduction of an integer $i$ modulo $h$  into  that interval.  Below all indexes are integers modulo $h$.

We first define  four permutation operations as below.
\begin{small}
\begin{itemize}
\item[] $\mu(A)_{i,j}=A_{i, i+j}$; $\zeta(A)_{i,j}=A_{i+j, j}$; \item[]$\phi(A)_{i,j}=A_{i, j+1}$; $\pi(A)_{i,j}=A_{i+1,j}$.
\end{itemize}
\end{small}
We can see that $\phi$ and $\pi$ are actually shifts of the columns and rows of the matrix, respectively.  Given two $(h\times h)$-dimensional square matrices $A$ and $B$, the  multiplication of $A$ and $B$ can be parsed as
\begin{small}
 \begin{equation}
 \label{eq2}
\begin{split}
A\cdot B=\sum_{k=0}^{h-1}(\phi^{k}\circ \mu(A) )\odot(\pi^{k}\circ \zeta (B)).
\end{split}
\end{equation}
\end{small}
The correctness of Eq.(\ref{eq2}) is shown as follows by calculating the components of the matrix index $(i, j)$.

\begin{small}
 \begin{equation}
\begin{split}
\sum_{k=0}^{h-1}(\phi^{k}\circ \mu(A) )_{i,j}\cdot(\pi^{k}\circ \zeta (B))_{i,j}&=\sum_{k=0}^{h-1}\mu(A)_{i,j+k}\cdot \zeta(B)_{i+k,j}\\
&=\sum_{k=0}^{h-1}A_{i, i+j+k} \cdot B_{i+j+k, j}\\
&=\sum_{k=0}^{h-1}A_{i,k}\cdot B_{k,j}=(A\cdot B)_{i,j}.
\end{split}
\end{equation}
\end{small}

Note that while a single $\mu(A)_{i,j+k}\cdot \zeta(B)_{i+k,j}=A_{i, i+j+k} \cdot B_{i+j+k, j}$  is not equal to $A_{i, k} \cdot B_{k,j}$, it is easy to deduce that $\sum_{k=0}^{h-1}A_{i, i+j+k} \cdot B_{i+j+k, j}=\sum_{k=0}^{h-1}A_{i,k}\cdot B_{k,j}=(A\cdot B)_{i,j}$.   To be precise, given $i$ and $j$, $\sum_{k=0}^{h-1}A_{i, i+j+k} \cdot B_{i+j+k, j}=\sum_{t=(i+j)}^{h-1+i+j}A_{i,t}\cdot B_{t,j}$, where we set $t=i+j+k$. Then, we have $\sum_{t=(i+j)}^{h-1+i+j}A_{i,t}\cdot B_{t,j}=\sum_{t=0}^{h-1}A_{i,t}\cdot B_{t,j}$ since all the indexes are considered as integers modulo $h$. Therefore, $\sum_{t=0}^{h-1}A_{i,t}\cdot B_{t,j}=\sum_{k=0}^{h-1}A_{i,k}\cdot B_{k,j}$.

We observe that Eq.(\ref{eq2}) consists of permutation and multiplication of element components between matrix entries. Intuitively, we can evaluate it using the operations (shown in \textbf{Algorithm}~\ref{algorithm 2}) provided by CKKS for packed ciphertexts. However, since the matrix representation $U$ usually has $n=h^2$ nonzero diagonal vectors, if we directly use \textbf{Algorithm}~\ref{algorithm 2} to evaluate $A\mapsto \phi^{k}\circ \mu(A)$ and $B\mapsto \pi^{k}\circ \zeta (B)$ for $1\leq k <h$, each of them requires
 rotations  with the complexity of $O(h^2)$. As a result,  the total complexity is $O(h^3)$. To alleviate this,  we design a new method to substantively improve its efficiency.

\subsection{Matrix Encoding}
 We introduce an encoding method that converts a matrix into a vector.  Given a vector $\mathbf{a}=(a_k)_{0\leq k<n}$, where $n=h^2$,  the encoding map $\iota: R^{n}\rightarrow R^{h\times h}$ is shown as below.

\begin{small}
 \begin{equation}
\begin{split}
\iota : \mathbf{a}\mapsto A=(a_{h\cdot i+j})_{0\leq i, j<h}.
\end{split}
\end{equation}
\end{small}
This encoding method makes the vector $\mathbf{a}$ essentially an ordered concatenation of the rows of the matrix $A$. As a result, $\iota(\cdot)$ is isomorphic of addition, which means that matrix addition operations are equivalent to the same operations between the corresponding original vectors. Therefore, the matrix addition can be calculated homomorphically in the SIMD environment. The constant multiplication operations can also be performed homomorphically. In this paper, we use $\iota(\cdot)$ to identify two spaces $ R^{n}$ and $R^{h\times h}$. For example, we say that a ciphertext  is  the encryption of $A$ if  $\mathbf{a}=\iota^{-1}(A)$.

\subsection{Tweaks of Permutation}
From the definition of matrix encoding,  permutation on an $(h\times h)$-dimensional matrix can be regarded as a linear transformation $\mathcal{T}: R^{n}\rightarrow R^{n}$, where $n=h^2$.  In
general, its matrix representation $U\in\{0, 1\}^{n\times n}\subset R^{n\times n}$ has $n$ nonzero diagonal vectors. Therefore, as presented in Sections~\ref{Preliminary Knowledge} and ~\ref{Permutation for matrix multiplication}, if we directly use \textbf{Algorithm}~\ref{algorithm 2} to evaluate $A\mapsto \phi^{k}\circ \mu(A)$ and $B\mapsto \pi^{k}\circ \zeta (B)$ for $1\leq k <h$, each of them requires
 rotations  with the complexity of $O(h^2)$. The total complexity will be $O(h^3)$. To  alleviate this,
 based on Eq.(\ref{eq2}) and our matrix encoding map, we provide a tweak method for matrix permutation to reduce the complexity from $O(h^3)$ to $O(h)$. Specifically, for four permutation operations ($\mu$, $\zeta$, $\phi$, and $\pi$) on the matrix, we use $U^{\mu}$, $U^{\zeta}$, $V$ and $P$ to indicate the matrix representations corresponding to these  permutations, respectively.  $U^{\mu}$, $U^{\zeta}$ for permutations  $\mu$ and $\zeta$ can be parsed as below (readers can refer to the example in Figure~\ref{Examples} for ease of understanding).
 \begin{small}
\begin{equation}
 U_{h\cdot i+j, t}^{\mu}=\left\{
\begin{aligned}
&1 \quad \mathsf{if}\; t=h\cdot i+[i+j]_h;\\
& 0 \quad \mathsf{otherwise};\\
\end{aligned}
\right.
\end{equation}
\begin{equation}
U_{h\cdot i+j, t}^{\zeta}=\left\{
\begin{aligned}
&1 \quad \mathsf{if}\; t=h\cdot[i+j]_h+j,\\
& 0 \quad \mathsf{otherwise};\\
\end{aligned}
\right.
\end{equation}
\end{small}
where $0\leq i,j<h$ and $0\leq t<h^2$. Similarly,  for $1\leq k<h$,  the  matrix representations of  $\phi^{k}$ and $\pi^{k}$ (i.e., $V^{k}$ and $P^{k}$) can be denoted as below.

\begin{small}
\begin{equation}
 V_{h\cdot i+j, t}^{k}=\left\{
\begin{aligned}
&1 \quad \mathsf{if}\; t=h\cdot i+[j+k]_h;\\
& 0 \quad \mathsf{otherwise};\\
\end{aligned}
\right.
\end{equation}
\begin{equation}
P_{h\cdot i+j, t}^{k}=\left\{
\begin{aligned}
&1 \quad \mathsf{if}\; t=h\cdot[i+k]_h+j;\\
& 0 \quad \mathsf{otherwise};\\
\end{aligned}
\right.
\end{equation}
\end{small}
where $0\leq i,j<h$ and $0\leq t<h^2$. Reviewing Eq.(\ref{eq1}), we use the diagonal decomposition of matrix representation to perform multiplication with  encrypted vectors. Hence, we can count the number of nonzero diagonal vectors in $U^{\mu}$, $U^{\zeta}$, $V$, and $P$ to evaluate the complexity.  For simplicity, we  use $\mathbf{u}_t$ to represent the $t$-th diagonal vector of a matrix $U$, and  identify $\mathbf{u}_{h^2-t}$ with $\mathbf{u}_{-t}$. For matrix $U^{\mu}$, we can  observe that it has exactly $(2h-1)$ nonzero diagonal vectors, denoted by $\mathbf{u}_k^{\mu}$ for $k \in \mathbb{Z}\cap (-h, h)$. There are $h$ nonzero diagonal vectors in $U^{\zeta}$, because each $t$-th diagonal vector in $U^{\zeta}$ is nonzero if and only if $t$ is divisible by the integer $h$. For each matrix $V^k$, $1\leq k<h$, it has only two nonzero diagonal vectors $\mathbf{v}_k$ and $\mathbf{v}_{k-h}$. Similarly,  for each matrix $P^{k}$, it has only one nonzero diagonal vector $\mathbf{p}_{h\cdot k}$. Therefore, we only need  rotation operations of $O(h)$ complexity to perform permutation ${\mu}$ and ${\zeta}$, and $O(1)$ complexity for both  $\phi^{k}$ and $\pi^{k}$ where $1\leq k<h$.
\begin{table*}[!]
\centering
\normalsize
\begin{tabular}{|p{18.2cm}|}
\hline \\
\textbf{Setup:} Given two ciphertexts  $[\mathbf{A}]_{pk}$ and $[\mathbf{B}]_{pk}$  that are the encryption forms of  two  $(3\times 3)$-dimensional matrix matrices $A$ and $B$ (shown below), respectively,  we now describe how to efficiently evaluate their homomorphic matrix multiplication.
\begin{small}
$\centering {\begin{matrix}
A=\begin{bmatrix}
a_0 & a_1&a_2 \\
a_3 &a_4 & a_5 \\
a_6 &a_7 &a_8
\end{bmatrix};&   B =\begin{bmatrix}
b_0 & b_1&b_2 \\
b_3 &b_4 & b_5 \\
b_6 &b_7 &b_8
\end{bmatrix}
\end{matrix}}$,
\end{small}
where the vector representations of $A$ and $B$ are $\mathbf{a}=[a_0, a_1, a_2, a_3, a_4, a_5, a_6, a_7, a_8]$ and  $\mathbf{b}=[b_0, b_1, b_2, b_3, b_4, b_5, b_6, b_7, b_8]$, respectively.\\
 \textbf{Step 1-1}: From $A$ and $B$, we first compute $U^{\mu}$, $U^{\zeta}$, $V=\{V^1, V^2\}$ and $P=\{P^1, P^2\}$ based on Eqn.(5)-(8) as follows. \\
$$\centering \begin{small} {\begin{matrix}
U^{\mu}=\begin{bmatrix}
1 & 0 &0& 0& 0&0&0&0&0 \\
0 & 1 &0& 0& 0&0&0&0&0 \\
0 & 0 &1& 0& 0&0&0&0&0 \\
0 & 0 &0& 0& 1&0&0&0&0 \\
0 & 0 &0& 0& 0&1&0&0&0 \\
0 & 0 &0& 1& 0&0&0&0&0 \\
0 & 0 &0& 0& 0&0&0&0&1 \\
0 & 0 &0& 0& 0&0&1&0&0 \\
0 & 0 &0& 0& 0&0&0&1&0
\end{bmatrix}; U^{\zeta}=\begin{bmatrix}
1 & 0 &0& 0& 0&0&0&0&0 \\
0 & 0 &0& 0& 1&0&0&0&0 \\
0 & 0 &0& 0& 0&0&0&0&1 \\
0 & 0 &0& 1& 0&0&0&0&0 \\
0 & 0 &0& 0& 0&0&0&1&0 \\
0 & 0 &1& 0& 0&0&0&0&0 \\
0 & 0 &0& 0& 0&0&1&0&0 \\
0 & 1 &0& 0& 0&0&0&0&0 \\
0 & 0 &0& 0& 0&1&0&0&0 \\
\end{bmatrix} ;
V^1=\begin{bmatrix}
0 & 1 &0& 0& 0&0&0&0&0 \\
0 & 0 &1& 0& 0&0&0&0&0 \\
1 & 0 &0& 0& 0&0&0&0&0 \\
0 & 0 &0& 0& 1&0&0&0&0 \\
0 & 0 &0& 0& 0&1&0&0&0 \\
0 & 0 &0& 1& 0&0&0&0&0 \\
0 & 0 &0& 0& 0&0&0&1&0 \\
0 & 0 &0& 0& 0&0&0&0&1 \\
0 & 0 &0& 0& 0&0&1&0&0
\end{bmatrix}
\end{matrix}}
\end{small}$$
 $$\centering \small {\begin{matrix}V^2=\begin{bmatrix}
0 & 0 &1& 0& 0&0&0&0&0 \\
1 & 0 &0& 0& 0&0&0&0&0 \\
0 & 1 &0& 0& 0&0&0&0&0 \\
0 & 0 &0& 0& 0&1&0&0&0 \\
0 & 0 &0& 1& 0&0&0&0&0 \\
0 & 0 &0& 0& 1&0&0&0&0 \\
0 & 0 &0& 0& 0&0&0&0&1 \\
0 & 0 &0& 0& 0&0&1&0&0 \\
0 & 0 &0& 0& 0&0&0&1&0
\end{bmatrix};
P^1=\begin{bmatrix}
0 & 0 &0& 1& 0&0&0&0&0 \\
0 & 0 &0& 0& 1&0&0&0&0 \\
0 & 0 &0& 0& 0&1&0&0&0 \\
0 & 0 &0& 0& 0&0&1&0&0 \\
0 & 0 &0& 0& 0&0&0&1&0 \\
0 & 0 &0& 0& 0&0&0&0&1 \\
1 & 0 &0& 0& 0&0&0&0&0 \\
0 & 1 &0& 0& 0&0&0&0&0 \\
0 & 0 &1& 0& 0&0&0&0&0
\end{bmatrix};
 P^2=\begin{bmatrix}
0 & 0 &0& 0& 0&0&1&0&0 \\
0 & 0 &0& 0& 0&0&0&1&0 \\
0 & 0 &0& 0& 0&0&0&0&1 \\
1 & 0 &0& 0& 0&0&0&0&0 \\
0 & 1 &0& 0& 0&0&0&0&0 \\
0 & 0 &1& 0& 0&0&0&0&0 \\
0 & 0 &0& 1& 0&0&0&0&0 \\
0 & 0 &0& 0& 1&0&0&0&0 \\
0 & 0 &0& 0& 0&1&0&0&0
\end{bmatrix}
\end{matrix}}$$
 We  securely compute  $U^{\mu}\cdot \mathbf{a}$. Based on Eqn.(9),  we have
$\centering {\begin{matrix}
U^{\mu}\cdot \mathbf{a}=[a_0, a_1, a_2, a_4, a_5, a_3, a_8, a_6, a_7]\overset{\iota(\mathbf{a})}{=}\begin{bmatrix}
a_0 & a_1&a_2 \\
a_4 &a_5 & a_3 \\
a_8 &a_6 &a_7
\end{bmatrix}
\end{matrix}}$,
where  $U^{\mu}$ has exactly $(2\times 3-1)=5$ nonzero diagonal vectors (based on Eqn.(10) and (11)) , denoted by $\mathbf{u}_k^{\mu}$ for $k \in \mathbb{Z}\cap (-3, 3)$. Specifically,  $\mathbf{u}_{-2}^{\mu}=[0, 0, 0, 0, 0, 1, 0, 0, 0]$, $\mathbf{u}_{-1}^{\mu}=[0, 0, 0, 0, 0, 0, 0, 1, 1]$, $\mathbf{u}_{0}^{\mu}=[1, 1, 1, 0, 0, 0, 0, 0, 0]$, $\mathbf{u}_{1}^{\mu}=[0, 0, 0, 1, 1, 0, 0, 0, 0]$, and $\mathbf{u}_{2}^{\mu}=[0, 0, 0, 0, 0, 0, 1, 0, 0]$. Then, we can get the ciphertext of $U^{\mu}\cdot \mathbf{a}$, denoted by $[\mathbf{A}^{(0)}]_{pk}$, based on Eqn.(12).\\
\textbf{Step 1-2}: We securely compute  $U^{\zeta}\cdot \mathbf{b}$. Based on Eqn.(13), we have
$\centering {\begin{matrix}
U^{\zeta}\cdot \mathbf{b}=[b_0, b_4, b_8, b_3, b_7, b_2, b_6, b_1, b_5]\overset{\iota(\mathbf{b})}{=}\begin{bmatrix}
b_0 & b_4&b_8 \\
b_3 &b_7 & b_2 \\
b_6 &b_1 &b_5
\end{bmatrix}
\end{matrix}}$,
where  $U^{\zeta}$ has exactly $h=3$ nonzero diagonal vectors, denoted by $\mathbf{u}_{3\cdot k}^{\zeta}$, for $0\leq k<3$. Specifically,  $\mathbf{u}_{0}^{\zeta}=[1, 0, 0, 1, 0, 0, 1, 0, 0]$, $\mathbf{u}_{3}^{\zeta}=[0, 1, 0, 0, 1, 0, 0, 1, 0]$, $\mathbf{u}_{6}^{\zeta}=[0, 0, 1, 0, 0, 1, 0, 0, 1]$. Then, we can get the ciphertext of $U^{\zeta}\cdot \mathbf{b}$, denoted by $[\mathbf{B}^{(0)}]_{pk}$, based on Eqn.(14).\\
\textbf{Step 2}: This step is used to securely perform column and row shifting operations on $\mu(A)$ and $\zeta(B)$ respectively. Specifically, for each column shifting matrix $V^k$, $1\leq k<3$, it has two nonzero diagonal vectors $\mathbf{v}_k$ and $\mathbf{v}_{k-h}$ (based on Eqn.(15) and (16)). Hence, the nonzero diagonal vectors in $V^1$ are $\mathbf{v}_{1}=[1, 1, 0, 1, 1, 0, 1, 1, 0]$ and $\mathbf{v}_{-2}=[0, 0, 1, 0, 0, 1, 0, 0, 1]$, and  the nonzero diagonal vectors in $V^2$ are $\mathbf{v}_{2}=[1, 0, 0, 1, 0, 0, 1, 0, 0]$ and  $\mathbf{v}_{-1}=[0, 1, 1, 0, 1, 1, 0, 1, 1]$. Similarly, for each row shifting matrix $P^{k}$, it has only one nonzero diagonal vector $\mathbf{p}_{3\cdot k}$. Then the nonzero diagonal vector in $P^1$ is $\mathbf{p}_{3}=[1, 1, 1, 1, 1, 1, 1, 1, 1]$ and  the nonzero diagonal vector in $P^2$ are $\mathbf{p}_{6}=[1, 1, 1, 1, 1, 1, 1, 1, 1]$.  Based on this,
we can obtain the ciphertexts $[\mathbf{A}^{(1)}]_{pk}$, $[\mathbf{A}^{(2)}]_{pk}$,  $[\mathbf{B}^{(1)}]_{pk}$, and $[\mathbf{B}^{(2)}]_{pk}$  of the matrix $\phi^{1}\circ \mu(A)$,  $\phi^{2}\circ \mu(A)$, $\pi^{1}\circ \zeta (B)$, and $\pi^{2}\circ \zeta (B)$, respectively, where
\begin{small}
$$\centering {\begin{matrix}
\phi^{1}\circ \mu(A)=\begin{bmatrix}
a_1 & a_2&a_0 \\
a_5 &a_3 & a_4 \\
a_6 &a_7 &a_8
\end{bmatrix};  \phi^{2}\circ \mu(A) =\begin{bmatrix}
a_2 & a_0&a_1 \\
a_3 &a_4 & a_5 \\
a_7 &a_8 &a_6
\end{bmatrix}; \pi^{1}\circ \zeta (B)=\begin{bmatrix}
b_3 & b_7&b_2 \\
b_6 &b_1 & b_5 \\
b_0 &b_4 &b_8
\end{bmatrix}; \pi^{2}\circ \zeta (B)=\begin{bmatrix}
b_6 & b_1&b_5 \\
b_0 &b_4 & b_8 \\
b_3 &b_7 &b_2
\end{bmatrix}
\end{matrix}}$$
\end{small}\\
\textbf{Step 3}: For $0\leq k<3$, we compute the element-wise multiplication between $[\mathbf{A}^{(k)}]_{pk}$ and $[\mathbf{B}^{(k)}]_{pk}$.  Then,  $[\mathbf{AB}]_{pk}$  is obtained as below.\\
\begin{footnotesize}
$$\centering {\begin{matrix}
\begin{bmatrix}
a_0 & a_1&a_2 \\
a_3 &a_4 & a_5 \\
a_6 &a_7 &a_8
\end{bmatrix}\cdot \begin{bmatrix}
b_0 & b_1&b_2 \\
b_3 &b_4 & b_5 \\
b_6 &b_7 &b_8
\end{bmatrix}= \begin{bmatrix}
a_0 & a_1&a_2 \\
a_4 &a_5 & a_3 \\
a_8 &a_6 &a_7
\end{bmatrix}\odot \begin{bmatrix}
b_0 & b_4&b_8 \\
b_3 &b_7 & b_2 \\
b_6 &b_1 &b_5
\end{bmatrix}+\begin{bmatrix}
a_1 & a_2&a_0 \\
a_5 &a_3 & a_4 \\
a_6 &a_7 &a_8
\end{bmatrix}\odot \begin{bmatrix}
b_3 & b_7&b_2 \\
b_6 &b_1 & b_5 \\
b_0 &b_4 &b_8
\end{bmatrix}+\begin{bmatrix}
a_2 & a_0&a_1 \\
a_3 &a_4 & a_5 \\
a_7 &a_8 &a_6
\end{bmatrix}\odot \begin{bmatrix}
b_6 & b_1&b_5 \\
b_0 &b_4 & b_8 \\
b_3 &b_7 &b_2
\end{bmatrix}
\end{matrix}}$$
\end{footnotesize}\\
\hline
\end{tabular}
\vspace{5pt}
\caption{{{Homomorphic multiplication of two $3\times 3$-dimensional matrices }}}
\label{Examples}
\end{table*}
\renewcommand\tablename{TABLE}
\renewcommand \thetable{\Roman{table}}
\setcounter{table}{0}
\setcounter{figure}{2}

\subsection{Homomorphic Matrix Multiplication}
\label{Homomorphic matrix multiplication}
Given two ciphertexts  $[\mathbf{A}]_{pk}$ and $[\mathbf{B}]_{pk}$  that are the encryption forms of  two  $(h\times h)$-dimensional matrix matrices $A$ and $B$, respectively,  we now describe how to efficiently evaluate homomorphic matrix multiplication between them.

\textbf{Step 1-1}: We  perform a linear transformation on the ciphertext $[\mathbf{A}]_{pk}$  under the guidance of permutation $U^{\mu}$ (\textbf{Step 1-1} in Figure~\ref{Examples}).  As described above, $U^{\mu}$ has exactly $(2h-1)$ nonzero diagonal vectors, denoted by $\mathbf{u}_k^{\mu}$ for $k \in \mathbb{Z}\cap (-h, h)$.  Then such a linear transformation can be denoted as
\begin{small}
 \begin{equation}
 \label{eq5}
\begin{split}
U^{\mu}\cdot \mathbf{a}=\sum_{-h<k<h}(\mathbf{u}_k^{\mu}\odot\mathbf{R}(\mathbf{a},k)),
\end{split}
\end{equation}
\end{small}
where $\mathbf{a}=\iota^{-1}(A)\in R^{n}$ is the vector representation of $A$.  If $k\geq 0$, the $k$-th diagonal vector can be computed as
\begin{small}
\begin{equation}
 \mathbf{u}_k^{\mu}[t]=\left\{
\begin{aligned}
&1 \quad \mathsf{if}\; 0\leq t-h\cdot k< (h-k);\\
& 0 \quad \mathsf{otherwise},\\
\end{aligned}
\right.
\end{equation}
\end{small}
where $\mathbf{u}_k^{\mu}[t]$ represents the $t$-th component of  $\mathbf{u}_k^{\mu}$. Similarly, if
\begin{small}
$k< 0$, $\mathbf{u}_k^{\mu}$ is computed as
\begin{equation}
\mathbf{u}_k^{\mu}[t]=\left\{
\begin{aligned}
&1 \quad \mathsf{if}\; -k\leq t-(h+k)\cdot h <h;\\
& 0 \quad \mathsf{otherwise},\\
\end{aligned}
\right.
\end{equation}
\end{small}
As a result, Eq.(\ref{eq5}) can be securely computed as
\begin{small}
 \begin{equation}
\begin{split}
\sum_{-h<k<h}\mathbf{Mul}_{pt}(\mathbf{Rot}([\mathbf{A}]_{pk}, k), \mathbf{u}_k^{\mu}),
\end{split}
\end{equation}
\end{small}
where we get the ciphertext of $U^{\mu}\cdot \mathbf{a}$, denoted as $[\mathbf{A}^{(0)}]_{pk}$.  We observe that the computation cost is about $2h$ rotations, constant multiplications and additions.

\textbf{Step 1-2}: This step is to   perform the linear transformation on the ciphertext $[\mathbf{B}]_{pk}$  under the guidance of permutation $U^{\zeta}$ (\textbf{Step 1-2} in Figure~\ref{Examples}).  Since $U^{\zeta}$ has  $h$ nonzero diagonal vectors,   this process  can be denoted as
\begin{small}
 \begin{equation}
\label{eq7}
\begin{split}
U^{\zeta}\cdot \mathbf{b}=\sum_{0\leq k<h}(\mathbf{u}_{h\cdot k}^{\zeta}\odot\mathbf{R}(\mathbf{b},h\cdot k)),
\end{split}
\end{equation}
\end{small}
where $\mathbf{b}=\iota^{-1}(B)\in R^{n}$, $\mathbf{u}_{h\cdot k}^{\zeta}$ is the $(h\cdot k)$-th diagonal vector of the matrix $U^{\zeta}$. We observe that for any $0\leq k<h$, $\mathbf{u}_{h\cdot k}^{\zeta}$  is a non-zero vector because its $(k+h\cdot i)$-th element is non-zero for $0\leq i<h$, and zero for all other entries. Therefore, Eq.(\ref{eq7}) can be securely computed as
\begin{small}
 \begin{equation}
\begin{split}
\sum_{0\leq k<h}\mathbf{Mul}_{pt}(\mathbf{Rot}([\mathbf{B}]_{pk}, h\cdot k), \mathbf{u}_{h\cdot k}^{\zeta}),
\end{split}
\end{equation}
\end{small}
where we get the ciphertext of $U^{\zeta}\cdot \mathbf{b}$, denoted as $[\mathbf{B}^{(0)}]_{pk}$.  We observe that the computation cost is about $h$ rotations, constant multiplications and additions.

\textbf{Step 2}: This step is used to securely perform column and row shifting operations on $\mu(A)$ and $\zeta(B)$ respectively (\textbf{Step 2} in Figure~\ref{Examples}). Specifically, for each column shifting matrix $V^k$, $1\leq k<h$, it has only two nonzero diagonal vectors $\mathbf{v}_k$ and $\mathbf{v}_{k-h}$, which are computed as
\begin{small}
\begin{equation}
 \mathbf{v}_{k}[t]=\left\{
\begin{aligned}
&1 \quad \mathsf{if}\; 0\leq [t]_h < (h-k);\\
& 0 \quad \mathsf{otherwise},\\
\end{aligned}
\right.
\end{equation}

\begin{equation}
\mathbf{v}_{k-h}[t]=\left\{
\begin{aligned}
&1 \quad \mathsf{if}\; (h-k)\leq [t]_h<h;\\
& 0 \quad \mathsf{otherwise}.\\
\end{aligned}
\right.
\end{equation}
\end{small}
 By adding two ciphertexts $\mathbf{Mul}_{pt}(\mathbf{Rot}([\mathbf{A}^{(0)}]_{pk},  k), \mathbf{v}_{k})$ and  $\mathbf{Mul}_{pt}(\mathbf{Rot}([\mathbf{A}^{(0)}]_{pk},  k-h), \mathbf{v}_{k-h})$, we can obtain the ciphertext $[\mathbf{A}^{(k)}]_{pk}$ of the matrix $\phi^{k}\circ \mu(A)$.  Similarly, for each row shifting matrix $P^{k}$, it has only one nonzero diagonal vector $\mathbf{p}_{h\cdot k}$.  Then the encryption of $\pi^{k}\circ \zeta (B)$ can be computed as  $[\mathbf{B}^{(k)}]_{pk}\leftarrow \mathbf{Rot}([\mathbf{B}^{(0)}]_{pk},  h\cdot k)$.  The computation cost of this process is about $3h$ rotations, $2h$ constant multiplications and $d$ additions.

\textbf{Step 3}: For $0\leq k<h$, we now compute the element-wise multiplication of $[\mathbf{A}^{(k)}]_{pk}$ and $[\mathbf{B}^{(k)}]_{pk}$ (\textbf{Step 3} in Figure~\ref{Examples}).  Then,  the ciphertext  $[\mathbf{AB}]_{pk}$ of the product of $A$ and $B$ is finally obtained. The computation cost of this process is $h$ homomorphic multiplications and additions. In summary, the entire process of performing homomorphic matrix multiplication is described in \textbf{Algorithm} \ref{algorithm 3}.

\begin{algorithm}
\caption{ Homomorphic matrix  multiplication }
\label{algorithm 3}
\begin{algorithmic}[1]
\small
\REQUIRE \texttt{HE-MatMult} $([\mathbf{A}]_{pk}, [\mathbf{B}]_{pk})$
\STATE $[\mathbf{A}^{(0)}]_{pk}\leftarrow $ \texttt{HE-LinTrans} $([\mathbf{A}]_{pk}, U^{\mu})$
\STATE  $[\mathbf{B}^{(0)}]_{pk}\leftarrow $ \texttt{HE-LinTrans} $([\mathbf{B}]_{pk}, U^{\zeta})$
\FOR { $k=1$ to $h-1$}
\STATE    $[\mathbf{A}^{(k)}]_{pk}\leftarrow $ \texttt{HE-LinTrans} $([\mathbf{A}^{(0)}]_{pk}, V^{k})$
\STATE   $[\mathbf{B}^{(k)}]_{pk}\leftarrow $ \texttt{HE-LinTrans} $([\mathbf{B}^{(0)}]_{pk}, P^{k})$
\ENDFOR
\STATE   $[\mathbf{AB}]_{pk} \leftarrow \mathbf{Mul}_{ct}([\mathbf{A}^{(0)}]_{pk}, [\mathbf{B}^{(0)}]_{pk})$
\FOR {$k=1$ to $h-1$}
\STATE    $[\mathbf{AB}]_{pk}\leftarrow \mathbf{Add} ([\mathbf{AB}]_{pk}, \mathbf{Mul}_{ct}([\mathbf{A}^{(k)}]_{pk}, [\mathbf{B}^{(k)}]_{pk}))$
\ENDFOR
\STATE \textbf{return} $[\mathbf{AB}]_{pk}$
\end{algorithmic}
\end{algorithm}

\textit{Remark 3.1}: In general, the above homomorphic matrix multiplication requires a total of $5h$ additions, $5h$ constant multiplications and $6h$ rotations. We can further reduce the computation complexity by using the baby-step/giant-step algorithm \cite{coron2005new,jiang2018secure} (See Appendix~\ref{AP:the baby-step/giant-step algorithm} for technical details). This algorithm can be exploited to reduce the complexity of Steps 1-1 and 1-2. As a result,  {Table}~\ref{Complexity of Algorithm 3} summarizes the computation complexity required for each step in \textbf{Algorithm}~\ref{algorithm 3}.

\begin{table}[H]
\centering
\small
\caption{Complexity of algorithm 3}
\label{Complexity of Algorithm 3}
\begin{tabular}{c|c|c|c|c}
\Xhline{1pt}
\textbf{Step}&$Add$& $mul_{pt}$ &$Rot$ & $mul_{ct}$\\
\Xhline{1pt}
\textbf{1-1}&$2h$&$2h$&$3\sqrt{h}$&-\\
\textbf{1-2}&$h$&$h$&$2\sqrt{h}$&-\\
\textbf{2}&$2h$&$h$&$3h$&-\\
\textbf{3}&$h$&-&$-$&$h$\\
\Xhline{1pt}
\textbf{Total}&$6h$&$4h$&$3h+5\sqrt{h}$&$h$\\
\Xhline{1pt}
\end{tabular}
\vspace{-5pt}
\end{table}
\textit{{Remark 3.2}}: As described before, the  multiplication of $A$ and $B$ is parsed as
$A\cdot B=\sum_{k=0}^{h-1}(\phi^{k}\circ \mu(A) )\odot(\pi^{k}\circ \zeta (B))$. A simple way to calculate the product is to directly use \textbf{Algorithm}~\ref{algorithm 2}: we can evaluate $A\mapsto \phi^{k}\circ \mu(A)$ and $B\mapsto \pi^{k}\circ \zeta (B)$ for $1\leq k <h$. However, each of them requires $O(h^2)$ homomorphic rotation operations, which results in a total complexity of $O(h^3)$ \cite{gilad2016cryptonets}. Halevi \textit{et al.} \cite{halevi2015bootstrapping} introduce a matrix encoding method based on  diagonal decomposition. This method maps each diagonal vector into a separate ciphertext by arranging the matrix diagonally. As a result, it requires $h$ ciphertexts to represent a matrix, and each ciphertext is required to perform matrix-vector multiplication with the complexity of $O(h)$ rotations, resulting in a total computation  complexity of $O(h^2)$. Compared with these schemes, our strategy only needs a total computation complexity of $O(h)$ rotations to complete the homomorphic multiplication for two $(h\times h)$-dimensional matrices. We note that  POSEIDON \cite{sav2020poseidon} also proposes an ``alternating packing (AP) approach''  to achieve matrix  multiplication with a complexity approximated as $\max_{i\in [\mathbb{L}]}(\omega_i\times\log(h\times \omega_i))$, where  $\omega_i$ denotes the number of weights between layers $i$ and $i+1$.  However, the implementation of this method requires to generate a large number of copies of each element in the matrix (depending on the number of neurons in the neural network layer where the matrix is located), resulting in poor parallel computing performance (see Section~\ref{sec:PERFORMANCE EVALUATION} for more experimental comparison).

\textit{{Remark 3.3}}: We also give the methods of how to perform matrix transposition, rectangular matrix multiplication (i.e., calculating general matrix forms such as $R^{t\times h}\times R^{h\times h}\rightarrow R^{t\times h}$ or $R^{h\times h}\times R^{h\times t}\rightarrow R^{h\times t}$) and parallel matrix operations (using the idleness of the plaintext slots) under packed ciphertext. They follow the similar idea of the above homomorphic matrix multiplication. Readers can refer to Appendix \ref{APP:Matrix Transposition on Packed Ciphertexts}, \ref{APP:Rectangular Matrix Multiplication on Packed Ciphertexts} and \ref{APP:Parallel  Matrix Computation} for more technical details.

\section{Approximation for Sign Function}
\label{Approximation for Comparison Function}
In this section, we describe how to efficiently estimate the sign function, and then use the estimated function to approximate the formulas commonly used in neural network training, including \texttt{ReLU} and \texttt{max} functions.
\subsection{Notations}
 We first introduce some useful symbols. Specifically, all logarithms are base $2$ unless otherwise stated. $\mathbb{Z}$ and $\mathbb{R}$ represent the integer and  real number fields, respectively. For a finite set $M$, we use $U(M)$ to represent the  uniform distribution on $M$. Given a function $g$ defined in the real number field $\mathbb{R}$, and a compact set $I \subset \mathbb{R}$, we say that the infinity norm of $g$ on the set $I$ is defined as $||g||_{\infty, I}:=\max_{m\in I}|g(m)|$, where $|g(m)|$ means the absolute value of $g(m)$.   we use $g^{(k)}:=g\circ g\circ g\circ \cdots \circ g$  to indicate the $k$-times composition of $g$. Besides, the sign function is defined as below.
$$ sgn (m)=\left\{
\begin{aligned}
&1 \quad \mathsf{if}\; m>0;\\
& 0 \quad \mathsf{if}\; m=0;\\
-&1  \quad \mathsf{if}\; m<0.\\
\end{aligned}
\right.
$$
Note that $sgn(m)$ is a discontinuous function at the zero point, so the closeness of $g(m)$ and $sgn(m)$ should be carefully considered in the interval near the zero point. That is, we do not consider the small interval $(-\delta, \delta)$ near the zero point when measuring the difference between $g(m)$ and $sgn(m)$. We will prove that  for some $k_d>0$,  the infinity norm of $g_{d}^{(k)}(m)-sgn(m)$ is small than $2^{-\sigma}$ over $[-1, -\delta]\cup[\delta, 1]$ if $k>k_d$, where the definition of  $g_{d}(m)$ will be explained later.

Given $\sigma>0$ and $0<\delta<1$, we define a function $g_{d}^{(k)}(m)$ that is $(\sigma, \delta)$-\textit{close} to $sgn(m)$ on $[-1, 1]$ if it satisfies
\begin{small}
 \begin{equation}
\begin{split}
||g_{d}^{(k)}(m)-sgn(m)||_{\infty, [-1, -\delta]\cup[\delta, 1]}\leq 2^{-\sigma}.
\end{split}
\end{equation}
\end{small}
Similar to the previous work \cite{cheon2019numerical}, we assume that the input  is limited to a bounded interval $[0, 1]$, since for any $m\in [a_1, a_2]$, where $a_2> a_1$,  we can scale it down to $[0, 1]$ by mapping $m\mapsto (m-a_1)/(a_2-a_1)$. Hence, for simplicity, the domain of $sgn(m)$ we consider in this part is $[-1, 1]$.

\subsection{Composite Polynomial Approximation}
\label{Composite Polynomial Approximate of Sign Function}
As mentioned before, we use a  composite function to approximate  the sign function. This is advantageous, because a composite polynomial function $G$, namely $G=g\circ g \cdots \circ g$, can be calculated with the complexity of $O(\log(deg (G)))$,  while the computation complexity of calculating any  polynomial $G$ is at least $\Theta(\sqrt{deg (G)})$ \cite{paterson1973number}, where $deg (G)$ indicates the degree of $G$. To achieve this, our goal is to find such a $k$ that $g^{(k)}$ is close enough to $sgn(x)$ in the interval $[-1, -\delta]\cup[\delta, 1]$.

Our construction of such a function $g$ comes from the following key observations: for any $m_0\in[-1, 1]$, let $m_i$ be the $i$-th composite value of $g^{(i)}(m_0)$. Then,   we can easily estimate the behavior of $m_i$ through the graph of $g$. Based on this, we ensure that as $i$ increases, $m_i$ should be close to 1 when $m_0\in (0, 1]$, and close to $-1$ when $m_0\in  [-1, 0)$.  Besides, we formally identify three properties of $g$ as follow. First, $g$ should be an odd function so as to be consistent with the sign function. Second, $g(1)=1$ and $g(-1)=-1$. This setting makes $g^{(k)}(m)$ point-wise converge to $sgn(m)$, whose value is $\pm1$ for all $x\neq 0$. In other words, for some $m\in [-1, 1]$, $g^{(k)}(m)$ converges to a value $y$ when increasing with the value of $k$, which means $g(y)=g(\lim_{k\rightarrow \infty}g^{(k)}(m))=\lim_{k\rightarrow \infty}g^{(k)}(m)=y$. Last, to accelerate the convergence of $g^{(k)}$ to the sign function, a satisfactory $g$ should be more concave in the interval $[0, 1]$ and more convex in the interval $[-1, 0]$. Moreover, the  derivative $g'$ of $g$ should have multiple roots at $1$ and $-1$ so as to increase the convexity. These properties are summarized as follows:\\
\textbf{Core Properties of $g$}: \\
Prop. I \;$g(-m)=-g(m)$ \quad  \;\;\;\;\; \;\; \quad \quad \; (Origin Symmetry)\\
Prop. II  \;$g(1)=1, g(-1)=-1$ \quad \quad \; (Convergence to $\pm 1$)\\
Prop. III \;$g'(m)=p (1-m)^{d}(1+m)^{d}$ \quad\; for some $p>0$ \\
  \rightline{(Fast convergence)}

Given a fixed $d\geq 1$,  a polynomial $g$ of degree $(2d+1)$ that satisfies the above three properties can be uniquely determined. We denote this polynomial as $g_d$, where the constant $p$ is indicated as $p_d$. Then, based on   Prop. I  and III, we have $g_d(m)=p_d\int_{0}^{m}(1-t^2)^{d}dt $, where the constant $p_d$ is also determined by Prop. II. To solve this integral formula $g_d(m)$, a common method is to transform the $(1-t^2)$ part of the integral formula with Trigonometric Substitutions, a typical technique which can convert formula $\int(1-t^2)^{d}dt$ to $\int(\cos t)^{3d}dt$. As a result, given   the following identity
 \begin{small}
 \begin{equation*}
\begin{split}
\int_{0}^{m}\cos^{n}t dt=\frac{1}{n}\cos^{n-1}m\cdot \sin m +\frac{n-1}{n}\int_{0}^{m}\cos^{n-2}t dt.
\end{split}
\end{equation*}
\end{small}
which holds for any  $n\geq 1$,  we have
 \begin{small}
 \begin{equation*}
g_{d}(m)=\sum_{i=0}^{i=d}\frac{1}{4^i}\cdot \begin{pmatrix}  2i \\  i \end{pmatrix}\cdot m(1-m^2)^{i}.
\end{equation*}
\end{small}

Therefore, we can compute $g_n$ as follows
\begin{itemize}
\item $g_1(m)=-\frac{1}{2}m^3+\frac{3}{2}m$.
\item $g_2(m)=\frac{3}{8}m^5-\frac{10}{8}m^3+\frac{15}{8}m$.
\item $g_3(m)=-\frac{5}{16}m^7+\frac{21}{16}m^5-\frac{35}{16}m^3+\frac{35}{16}m$.
\item $g_4(m)=\frac{35}{128}m^9-\frac{180}{128}m^7+\frac{378}{128}m^5-\frac{420}{128}m^3+\frac{315}{128}m$.
\end{itemize}
Since $\begin{pmatrix}  2i \\  i \end{pmatrix}= 2 \cdot \begin{pmatrix}  2i-1 \\  i-1 \end{pmatrix}$ is divisible by 2 for $i\geq 1$, each coefficient of $g_d$ can be represented as $n/2^{2d-1}$ for $n\in \mathbb{Z}$, which can be inferred by simply using Binomial Theorem  for the coefficients in $g_d(m)$.\\
\textbf{Size of  constant $p_d$}:  The constant $p_d$ is crucial for $g_d^{(k)}$ to converge to the sign function. Informally, since the coefficient term of $m$  is exactly $p_d$, we can regard $g_d(m)$ as $g_d(m)\simeq p_d\cdot m$ for small $m$. Further we have $1-g_d(m)\simeq 1-p_d\cdot m\simeq (1-m)^{p_d}$. For simplicity, we can obtain $p_d$  as follows:
 \begin{small}
 \begin{equation*}
\sum_{i=0}^{i=d}\frac{1}{4^i}\cdot \begin{pmatrix}  2i \\  i\end{pmatrix},
\end{equation*}
\end{small}
which can be simplified with Lemma~\ref{lemma2}.
\begin{lemma}
\label{lemma2}
It holds that $p_d=\sum_{i=0}^{i=d}\frac{1}{4^i}\cdot \begin{pmatrix}  2i \\  i \end{pmatrix}=\frac{2d+1}{4^d}\begin{pmatrix}  2d \\  d \end{pmatrix}$.
\end{lemma}
\begin{proof}
Please refer to Appendix~\ref{AP:pof LE3}.
\end{proof}
\subsection{Analysis on the Convergence of $g_d^{(k)}$}
We now analyze the convergence of $g_d^{(k)}$ to the sign function as $k$ increases. To be precise, we provide a lower bound on $k$, under which $g_d^{(k)}$ is $(\sigma, \delta)$-\textit{close} to the sign function. To accomplish this, we first give two lower bounds about $1-g_d(m)$ as shown below.
\begin{lemma}
\label{lemma3}
It holds that $0\leq 1-g_d(m)\leq (1-m)^{p_d}$ for $m\in [0, 1]$.
\end{lemma}
\begin{lemma}
\label{lemma4}
It holds that $0\leq 1-g_d(m)\leq 2^d\cdot (1-m)^{d+1}$ for $m\in [0, 1]$, where the value of $m$ is close to 1.
\end{lemma}

\begin{proof}
Please refer to Appendix~\ref{AP:pof LE4} and~\ref{AP:pof LE5}.
\end{proof}

\begin{theorem}
\label{the3.5}
If $k\geq \frac{1}{\log p_d}\cdot \log (1/\delta)+\frac{1}{\log (d+1)}\cdot \log (\sigma-1)+O(1)$, then $g_d^{(k)}(m)$ is an $(\sigma, \delta)$-close polynomial to $sgn(x)$ over $[-1, 1]$.
\end{theorem}
\begin{proof}
Here we only consider the case where the input of $g_d^{(k)}$ is non-negative, since $g_d^{(k)}$ is an odd function.  We use Lemma ~\ref{lemma3} and Lemma ~\ref{lemma4} to analyze the lower bound of $k$ when $g_d^{(k)}$  converges to $(\sigma, \delta)$-close polynomial to $sgn(x)$.  Note that when the value of $m$ is close to $0$, Lemma ~\ref{lemma3} is tighter than Lemma ~\ref{lemma4} but vice verse when the value of $x$ is close to 1.  To obtain a tight lower bound on $k$, we decompose the proof into the following two steps, each of which  applies Lemma ~\ref{lemma3} and \ref{lemma4}, separately.

\textit{Step 1}. We consider the case $m\in [\delta, 1]$ instead of $[-1, -\delta]\cup [\delta, 1]$, since $g_d^{(k)}$ is an odd function. Let $k_{\delta}=\lceil \frac{1}{\log(p_d)} \cdot \log (\log(\frac{1}{\gamma})/\delta)\rceil$ for some constant $0<\gamma<1$. Then, with Lemma~\ref{lemma3}, we have the following inequality for $m\in [\delta, 1]$.
\begin{small}
 \begin{equation*}
 \begin{split}
1-g_d^{k_\delta}(m)&\leq (1-m)^{p_{d}^{k_\delta}}\leq (1-\delta)^{\log(\frac{1}{\gamma}/\delta)}< (\frac{1}{e})^{\log(\frac{1}{\gamma})}<\gamma,
\end{split}
\end{equation*}
\end{small}
where $e$ indicates the Euler's constant.

\textit{Step 2}. Let $k_\sigma=\lceil  \frac{1}{\log(d+1)}\cdot \log((\sigma-1/\log(\frac{1}{2\gamma})))\rceil$.  With  Lemma~\ref{lemma4}, we have the following inequality for $m\in [\delta, 1]$.

\begin{small}
 \begin{equation*}
 \begin{split}
2\cdot(1-g_d^{(k_\delta+k_\sigma)}(m))&\leq (2\cdot (1-g_d^{k_\delta}(m)))^{(d+1)^{k_\sigma}}\\
&\leq (2\gamma)^{(d+1)^{k_\sigma}}\leq (2\gamma)^{{\sigma-1}/\log(\frac{1}{2\gamma})}\\
&=2^{-\sigma +1}.
\end{split}
\end{equation*}
\end{small}
Therefore,  $1-g_d^{(k)}(m)\leq 2^{-\sigma}$ for $m\in[\delta, 1]$, if $k\geq k_\delta+k_\sigma$.\\
\end{proof}
\textbf{Comparisons with existing works}. We compare the computation complexity of our method with existing approximation methods for the sign function, including the traditional Minmax based polynomial approximation method \cite{eremenko2007uniform} and the latest  work \cite{cheon2019numerical}. The results are shown in Table~\ref{Complexity of Each Approximation Method}. \textit{Paterson et al.}\cite{paterson1973number} have proven that when the input is within the interval $[-1, 1]$, the minimum degree of a $(\sigma, \delta)$-polynomial function  to approximate a sign function is $\Theta(\sigma/\delta)$. This means at least multiplications with the complexity of $\Theta(\log (1/\delta))+\Theta(\log \sigma)$ are required to complete the approximation of the sign function. Hence, our method achieves an optimality in asymptotic computation complexity. Other works, like \cite{cheon2019numerical} as one of the most advanced solutions for approximating the sign function, only achieve quasi-optimal computation complexity (see TABLE~\ref{comparisons with work116} in APPENDIX~\ref{Experimental comparisons} for more experimental comparisons).
\begin{table}[H]
\centering
\footnotesize
\caption{Complexity of Each Approximation Method}
\label{Complexity of Each Approximation Method}
\begin{tabular}{c|c|c|c}
\Xhline{1pt}
{Parameter}&MinMax Approx.\cite{eremenko2007uniform}&\cite{cheon2019numerical} &\textbf{Ours}\\
\Xhline{1pt}
$\log(\frac{1}{\delta})=\Theta(1)$&$\Theta(\sqrt{\sigma})$& $\Theta(\log^2\sigma)$& ${\Theta(\log\sigma)}$\\
\hline
$\log(\frac{1}{\delta})=\Theta(\sigma)$&$\Theta(\sqrt{\sigma}\cdot2^{\frac{\sigma}{2}})$& $\Theta(\sigma\cdot\log\sigma)$& ${\Theta(\sigma)}$\\
\hline
$\log(\frac{1}{\delta})=2^\sigma$&$\Theta(\sqrt{\sigma}\cdot2^{2^{\sigma-1}})$& $\Theta(\sigma\cdot2^\sigma)$& ${\Theta(2^\sigma)}$\\
\Xhline{1pt}
\end{tabular}
\end{table}

\subsection{Application to Max and Relu Functions}
Given two variables $a$ and $b$, the \texttt{max} function can be expressed as $\max(a,b)=\frac{a+b}{2}+\frac{|a-b|}{2}$.
The \texttt{ReLu} function  $f(x)=\max (0, x)$ can be considered as a special case of the \texttt{max} function.
Specifically, since $|m|=m\cdot sgn(m)$, as long as we give the approximate polynomial about $|m|$, we can directly get the approximate \texttt{max} function. Therefore,  $\max(a,b)$ can be evaluated by  computing $\frac{a+b}{2}+\frac{a-b}{2}\cdot g_d^{(k)}(a-b)$. The detailed algorithm is shown in \textbf{Algorithm}~\ref{algorithm 5}. We also provide the convergence rate to approximate the absolute  function $|m|$ with $m\cdot g_d^{(k)}(m)$ (See Theorem \ref{theh3.6}).
\begin{algorithm}
\caption{Approximation of the maximum function }
\label{algorithm 5}
\begin{algorithmic}[1]
\small
\REQUIRE \texttt{AppMax} $(a, b, d, k)$
\STATE $m\leftarrow a-b$, $y\leftarrow \frac{a+b}{2}$
\FOR{{$k=1$ to $k=n-1$}}
\STATE $m\leftarrow g_d(m)$
\ENDFOR
\STATE $y\leftarrow y+\frac{a-b}{2}\cdot m$
\STATE \textbf{return} $y$
\end{algorithmic}
\end{algorithm}
\begin{theorem}
\label{theh3.6}
If $k\geq\frac{1}{\log p_d}\cdot \log (\sigma-1)$, then the error of  $m\cdot g_d^{(k)}(m)$ compared with   $|m|$ over $[-1, 1]$ is bounded by $2^\sigma$.
\end{theorem}
\begin{proof}
This proof can be easily evolved from Theorem \ref{the3.5}. We omit it for brevity.

\end{proof}

\section{Implementation of \Name}
\label{Implementation of Hercules}
We now describe the technical details of implementing \Name, which  provides privacy-preserving federated neural network training. In particular, model parameters and users' data are encrypted throughout the execution process.  To achieve this, \Name exploits the  MCKKS as the underlying framework and relies on the packed  ciphertext technology to accelerate calculations. Besides,  approximation methods based on composite polynomials are used to approximate  \texttt{ReLU} and \texttt{max} functions, which facilitate the compatibility of HE with complex operations.

From a high-level view,  the implementation of \Name is composed of three phases: \texttt{Prepare}, \texttt{Local Training}, and \texttt{Aggregation}. As shown in \textbf{Algorithm}~\ref{algorithm 6}, we use $[\cdot]_{pk}$ to denote  the encrypted value and ${{\omega}_{j, i}^{k}}$ to represent the weight matrix of the $j$-th layer generated by $P_i$ at the $k$-th iteration. The global weight  matrix is denoted as ${{\omega}_{j}^{k}}$ without index $i$. Similarly, the local gradients  computed by user $P_i$ for each layer $j$ at the $k$-th iteration is denoted as ${\bigtriangledown {\omega}_{j, i}^{k}}$.
\renewcommand{\algorithmicrequire}{\textbf{Input:}}
\renewcommand{\algorithmicensure}{\textbf{Output:}}
\begin{algorithm}
\caption{\small{High-level of federated neural network training}}
\label{algorithm 6}
\begin{algorithmic}[1]
\footnotesize
\REQUIRE $\{x, y\}\in D_i\subseteq D$, for $i\in \{1, \cdots, N \}$
\ENSURE  Encrypted ${\omega}_{1}^{H}, {\omega}_{2}^{H}, \cdots, {\omega}_{\mathbb{L}}^{H}$

\texttt{Prepare}:
\STATE The cloud server $\mathcal{C}$ and every user ${P}_i$ agree on the parameters $\mathbb{L}$, $h_1, \cdots, h_\mathbb{L}$, $\eta$, $\varphi (\cdot)$, $H$ and  $\mathcal{B}$.
The cloud server $\mathcal{C}$ generates  its secret key and public key $\{sk', pk'\} \leftarrow \mathbf{SecKeyGen}(1^{\lambda})$.
\STATE  Each user ${P}_i$  generates $sk_i \leftarrow \mathbf{SecKeyGen}(1^{\lambda})$.
\STATE All users collectively generate $pk\leftarrow\mathbf{DKeyGen}(\{sk_i\})$.
\STATE Each user encodes its input as $\hat{X_i}$, $\hat{Y_i}$. \footnotemark{}
\STATE The cloud server $\mathcal{C}$ initializes ${[\mathbf{\omega}_{1}^{0}}]_{pk}, [\mathbf{{\omega}_{2}^{0}}]_{pk}, \cdots, [\mathbf{{\omega}_{\mathbb{L}}^{0}}]_{pk}$. Then, $\mathcal{C}$ broadcasts them  to all users.

\texttt{Local Training}:
\FOR{{$k=0$ to $k=H-1$}}
\STATE Each user ${P}_i$ computes $[\mathbf{\bigtriangledown {\omega}_{1, i}^{k}}]_{pk}, \cdots, [\mathbf{\bigtriangledown {\omega}_{\mathbb{L}, i}^{k}}]_{pk}$ and sends them to the cloud server.

\texttt{Aggregation}:
\FOR{{$j=1$ to $j=\mathbb{L}$}}
\STATE $\mathcal{C}$   computes $[\mathbf{{\bigtriangledown \omega}_{j}^{k}}]_{pk}=[\sum_{i=1}^{N}\mathbf{\bigtriangledown {\omega}_{j, i}^{k}}]_{pk}$.
\STATE $\mathcal{C}$   computes $[\mathbf{{\omega}_{j}^{k+1}}]_{pk}=[\mathbf{{\omega}_{j}^{k}}-\frac{\eta}{\mathcal{B}\times N}\mathbf{\bigtriangledown {\omega}_{j}^{k}}]_{pk}$ and broadcasts them to all users.
\ENDFOR
\ENDFOR
\end{algorithmic}
\end{algorithm}
\footnotetext{$\hat{X_i}$ and $\hat{Y_i}$ can be vectors composed of a single training sample, or a matrix composed of multiple samples. This depends on the size of a single sample and the value of the degree $\mathcal{N}$ of  the cyclotomic polynomial ring.}

1. \texttt{Prepare}: The cloud server $\mathcal{C}$ needs to agree with all users on the training hyperparameters, including the number $\mathbb{L}$ of layers in the model, the number $h_j$ of neurons in each hidden layer $j$, $j\in[\mathbb{L}]$, the learning rate $\eta$, the number $H$ of global iterations, the number $\mathcal{B}$ of local batches, the activation function $\varphi (\cdot)$ and its approximation. Then, $\mathcal{C}$  generates its own key pair $\{ sk', pk'\}$, and each user  ${P}_i$ generates $sk_i$ for $i\in[N]$. Besides, all users collectively generate $pk$. Finally,   $\mathcal{C}$ initializes ${[\mathbf{\omega}_{1}^{0}}]_{pk}, [\mathbf{{\omega}_{2}^{0}}]_{pk}, \cdots, [\mathbf{{\omega}_{\mathbb{L}}^{0}}]_{pk}$, and broadcasts them  to all users.

2. \texttt{Local Training}: Each user ${P}_i$  executes the mini-batch based SGD algorithm locally and obtains the encrypted local gradients $[\mathbf{\bigtriangledown {\omega}_{1, i}^{k}}]_{pk}, \cdots, [\mathbf{\bigtriangledown {\omega}_{\mathbb{L}, i}^{k}}]_{pk}$, where  $P_i$ is required to execute the forward and backward passes for $\mathcal{B}$ times to compute and aggregate the local gradients.  Then,  ${P}_i$ sends these local gradients to the cloud server $\mathcal{C}$.

3. \texttt{Aggregation}: After receiving all the local gradients from users, $\mathcal{C}$ updates the global model parameters by computing the averaged aggregated gradients.  In our system, training is stopped once the number of iterations reaches $H$. Therefore, after the last iteration, all users need to perform an additional ciphertext conversion operation, i.e., the $\mathbf{DKeySwitch}$ function (shown in Figure~\ref{MCKKS}), which enables  to convert  model $\mathbf{M}$ encrypted under the  public key $pk$ into $[\mathbf{M}]_{pk'}$ under the cloud server's public key $pk'$ without decryption, so that $\mathcal{C}$ can access the final model parameters.

Figure~\ref{Detailed description of the scheme} in Appendix~\ref{AP:Details Implementation of Hercules} presents the details of \Name implementation, which essentially executes \textbf{Algorithm}~\ref{algorithm 1} under the ciphertext.
This helps readers understand how the functions in MCKKS as well as our new matrix parallel multiplication technology are used in FL.

\noindent\textbf{Security of \Name}:
  We demonstrate that  \Name realizes the data and model privacy protection defined in Section~\ref{Threat Model and Privacy Requirements}, even under the collusion of up to $N-1$ users. This is inherited from the property of MCKKS \cite{mouchet2020multiparty}. We give the following \textbf{Theorem}~\ref{theorem1} and provide  the security proof (sketch). The core of our proof is that for any adversary, when only the input and output  of  passive malicious users in \Name are provided, there  exists a simulator  with Probabilistic Polynomial Time  computation ability, which can simulate the view of the adversary and make  the adversary  unable to distinguish the real view from the simulated one.
\begin{theorem}
\label{theorem1}
 \Name realizes the privacy protection of data and model parameters during the FL process, as long as its underlying MCKKS cryptosystem is secure.
\end{theorem}
\renewcommand{\proofname}{\indent Proof \rm(Sketch)}
\begin{proof}
  \Name  inherits the security attributes of the MCKKS cryptosystem  proposed by
   \textit{Mouchet et al.} \cite{mouchet2020multiparty}. Compared with the standard CKKS, the multiparty version constructs additional distributed cryptographic functions including $\mathbf{DKeyGen}(\cdot)$, $\mathbf{DDec}(\cdot)$, $\mathbf{DKeySwitch}(\cdot)$ and $\mathbf{DBootstrap}(\cdot)$. All of them have been proven secure against a passive-adversary model with up to $N-1$ colluding parties, under the assumption of the underlying NP  hard problem (i.e., RLWE problem\cite{rosca2018ring} ).  Here we give a sketch of the proof with the simulation paradigm of the real/ideal world.
Let us assume that a real-world simulator $\mathcal{S}$ simulates a computationally bounded adversary composed of $N-1$ users  colluding with each other. Therefore,  $\mathcal{S}$ can access all the inputs and outputs of these $N-1$ users. As mentioned earlier, the MCKKS guarantees the indistinguishability of plaintext under chosen plaintext attacks (i.e., CPA-Secure) even if collusion of $N-1$ users. This stems from the fact that the secret key used for encryption must be recovered with the participation of all users. Therefore, $\mathcal{S}$ can simulate the data sent by honest users by replacing the original plaintext with random messages. Then these random messages are encrypted and sent to the corresponding adversary. Due to the security of CKKS, the simulated view is indistinguishable from the real view to the adversary. Analogously,  the same argument proves that  \Name protects the privacy of the training model, because all model parameters are encrypted with CKKS, and the intermediate and final weights are always in ciphertext during the training process.
\end{proof}

\section{Performance Evaluation}
\label{sec:PERFORMANCE EVALUATION}
 We  experimentally evaluate the performance of  \Name in terms of classification accuracy, computation  communication and storage overhead. We compare \Name with POSEIDON \cite{sav2020poseidon}, which is consistent with our scenario and is also bulit on MCCKS.

\subsection{Experimental Configurations}
We implement the multi-party cryptography operations on the Lattigo lattice-based library \cite{lattigo}, which provides an optimized version of the MCKKS cryptosystem. All the experiments are performed on 10 Linux servers, each of which is equipped with Intel Xeon E5-2699v3 CPUs, 36 threads on 18 cores and 189 GB RAM. We make use of Onet \cite{onts} and build a distributed  system where the parties communicate over TCP with secure channels (TLS).
 We instantiate \Name with the number of users as $N=10$ and $N=50$, respectively.  For parameter settings, the  dimension of the cyclotomic polynomial ring in CKKS is set as $\mathcal{N}=2^{13}$ for the datasets with the dimension of input $h<32$ or $32\times32$ images, and $2^{14}$ for those with inputs $h>32$. The number of initial levels  $\mathcal{L}=6$.  We exploit $g_4(m)$ described in Section~\ref{Composite Polynomial Approximate of Sign Function} as the basic of compound polynomial to approximate the \texttt{ReLU} and \texttt{max} functions, where we require $\sigma=20$, $\delta=2^{-20}$.  For other continuous activation functions, such as sigmoid, we  use the traditional MinMax strategy to approximate it, since it has been proven that a small degree polynomial can fit a non-polynomial continuous function well within a small bounded error.

Consistent with POSEIDON \cite{sav2020poseidon}, we choose 7 public datasets (i.e., BCW \cite{lavanya2011analysis}, MNIST  \cite{deng2012mnist}, ESR \cite{ESR}, CREDIT \cite{CREDIT}, SVHN \cite{SVHN}, CIFAR-10, and CIFAR-100 \cite{CIFAR}) in our experiments, and  design 5 different neural network architectures trained on the corresponding datasets (See Appendices~\ref{AP:DATA}  and~\ref{AP:NNS} for more details of the datasets and models used in our experiments. Note that we train two models, CIFAR-10-N1 and CIFAR-10-N2, over the CIFAR-10 dataset  for comparison).

\subsection{ Model Accuracy}
We first discuss the model accuracy on different datasets when the number of users is 10 and 50 respectively. We choose the following three baselines for comparison. (1) Distributed: distributed  training in plaintext, which is in the plaintext form corresponding to \Name. The datasets are evenly distributed to all users to perform FL in a plaintext environment. (2) Local: local  training in plaintext, i.e., each user only trains the model on the local dataset. (3) POSEIDON \cite{sav2020poseidon}. We reproduce the exact algorithm designed in \cite{sav2020poseidon}.

\begin{table*}[!htbp]
\footnotesize
\centering
\caption{Model accuracy and training Cost with $N=10$ users}
\label{Model accuracy1}
\begin{tabular}{|c|c|c|c|c|c|c|c|c|c|c|}
\Xhline{1pt}
\multirow{3}*{Dataset}&\multicolumn{4}{c|}{Accuracy}&\multicolumn{4}{c|}{Training time (s)}& \multicolumn{2}{c|}{Communication cost (GB)}\\\cline{2-11}
&\multirow{2}*{Distributed}&\multirow{2}*{Local}&\multirow{2}*{POSEIDON}& \multirow{2}*{\Name}&  \multicolumn{2}{c|}{POSEIDON}& \multicolumn{2}{c|}{\Name}&\multirow{2}*{POSEIDON}& \multirow{2}*{\Name}\\\cline{6-9}
&&&&& One-GI& Total& One-GI& Total&&\\ \Xhline{1pt}
BCW&97.8\%&93.9\%&96.1\%&\textbf{97.7\%}&0.40&39.92&0.11&\textbf{11.09}&0.59&\textbf{0.59}\\
\hline
ESR&93.6\%&90.1\%&90.2\%&\textbf{93.3\%}&0.92&553.44&0.29&\textbf{172.95}&562.51&\textbf{3.52}\\ \hline
CREDIT&81.6\%& 79.6\%&80.2\% &\textbf{81.4\%} &0.33&163.07&0.13&\textbf{62.73}&7.32&\textbf{2.93}\\
\hline
MNIST&92.1\%&87.8\%&88.7\%&\textbf{91.8\%}&44.67&4467.25&1.54&\textbf{1540.43}&703.13&\textbf{17.58}\\
\Xhline{1pt}
\end{tabular}
\end{table*}
All the baselines are trained on the same network architecture and learning hyperparameters. The learning rate is adaptive to different schemes to obtain the best training accuracy\footnote{For example, approximating the activation function at a small interval usually requires a small learning rate to avoid divergence.}.

As shown in  Tables~\ref{Model accuracy1} and \ref{Model accuracy2}, we can obtain the following two observations. (1) Compared with local training, FL improves the accuracy of model training, especially with the participation of large-scale users. This is drawn from the comparison between the second and fourth columns of Table~\ref{Model accuracy2}. The reason is obvious: the participation of large-scale users has enriched the volume of training samples, and a more accurate model can be derived from such a fertile composite dataset. (2) Compared with distributed training, \Name has negligible loss in accuracy (less than $0.3\%$) and is obviously better than POSEIDON ($1\%$ to $4\%$ improvement). In POSEIDON, the non-continuous activation function (i.e, \texttt{ReLU}) is converted into a low-degree polynomial using a traditional approximation method based on the least square method. This is computationally efficient but inevitably brings a non-negligible precision loss.  However, given a small error bound, our approximation  based on the composite polynomial can approximate non-continuous functions with high-degree polynomials, but only requires the computation complexity of $O(\log(deg G))$, where $deg G$ is the degree of the  composite polynomial. Therefore, the  accuracy loss caused by the conversion of the activation function is very slight in \Name.

Note that the model accuracy can be further improved by increasing the number of iterations, but we use the same number of iterations for the convenience of comparison. To achieve the expected training accuracy, model training over CIFAR-100 usually requires a special network architecture (such as ResNet) and layers (batch normalization) due to the diversity of its labels. For the training simplicity, we choose a relatively simple network architecture, which is also the main reason for the relatively low training accuracy under CIFAR-10 and CIFAR-100. We leave the model training of more complex architectures and tasks as future work (See Appendix~\ref{AP:Learning Extensions}).
\subsection{Computation Overhead}
\begin{table*}[!htbp]
\footnotesize
\centering
\caption{Model accuracy and training cost with $N=50$ users}
\label{Model accuracy2}
\begin{threeparttable}
\begin{tabular}{|c|c|c|c|c|c|c|c|c|c|c|}
\Xhline{1pt}
\multirow{3}*{Dataset}&\multicolumn{4}{c|}{Accuracy}&\multicolumn{4}{c|}{Training time (hrs)}& \multicolumn{2}{c|}{Communication cost (GB)}\\\cline{2-11}
&\multirow{2}*{Distributed}&\multirow{2}*{Local}&\multirow{2}*{POSEIDON}& \multirow{2}*{\Name}&  \multicolumn{2}{c|}{POSEIDON}& \multicolumn{2}{c|}{\Name}&\multirow{2}*{POSEIDON}& \multirow{2}*{\Name}\\\cline{6-9}
&&&&& One-GI& Total& One-GI& Total&&\\ \Xhline{1pt}
SVHN&68.4\%&35.1\%&67.5\%&\textbf{68.2\%}&0.0013&24.15&0.0005&\textbf{8.78}&12656.25&\textbf{474.61}\\
\hline
CIFAR-10-N1&54.6\%&26.8\%&51.8\%&\textbf{54.3\%}&0.005&126.26&0.0016&\textbf{40.73}&61523.44&\textbf{2050.78}\\\hline
CIFAR-10-N2&63.6\% &28.0\%&60.1\% &\textbf{63.1\%} &0.0059&98.32&0.002&\textbf{33.33}&59062.5&\textbf{1968.75}\\
\hline
CIFAR-100&43.6\%&8.2\%&40.1\%&\textbf{43.4\%}&0.0069&363.11&0.0024&\textbf{126.52}&246796.88&\textbf{8226.56}\\
\Xhline{1pt}
\end{tabular}
\begin{tablenotes}
        \item[] \footnotesize{Note that \Name and POSEIDON produce a relatively high total communication overhead compared to Table~\ref{Model accuracy1}, which stems from the use of a larger number of global iterations over the above datasets (See Appendix~\ref{AP:NNS} for hyperparameter settings).}
      \end{tablenotes}
 \end{threeparttable}
\end{table*}
We further discuss the performance of \Name in terms of computation overhead. As shown in Tables~\ref{Model accuracy1} and \ref{Model accuracy2}, when the number of users is 10, the training time of \Name over BCW, ESR and CREDIT is less than 3 minutes, and the training time over MNIST is also less than 30 minutes. For $N=50$, to  train  specific model architectures over SVHN, CIFAR-10-N1, CIFAR-10-N2 and CIFAR-100,  the total cost of \Name  is  $8.78$ hours, $40.73$ hours, $33.3$ hours and $126.52$  hours, respectively. We also give the running time of one global iteration (One-GI), which can be used to estimate the time required to train these architectures under a larger number of global iterations.  Obviously, for the same model architecture and number of iterations, the execution time of \Name is far less than that of POSEIDON. This stems from the fast SIMD operation under our new matrix multiplication coding method (See Appendix~\ref{Microbenchmarks}
for the comparison of the microbenchmark costs of \Name and POSEIDON under various functionalities). Specifically, POSEIDON designs AP to achieve fast SIMD calculations. AP combines row-based and column-based packing, which means that the rows or columns of the matrix are vectorized and packed into a ciphertext. For the multiplication of two $(h\times h)$-dimensional matrices, the complexity of the homomorphic rotation operations required by AP is $\max_{i\in [\mathbb{L}]}(\omega_i\times\log(h\times \omega_i))$, where  $\omega_i$ denotes the number of weights between layers $i$ and $i+1$. For example, given $h=64$, $\max_{i\in [\mathbb{L}]}\omega_i=64$, AP roughly needs $768$ homomorphic rotation operations to realize the multiplication calculation of two $(64\times 64)$-dimensional matrices. For \Name, as shown in Table~\ref{Complexity of Algorithm 3}, the complexity required for the matrix multiplication is only $3\times 64+ 5\sqrt{64}=232$, which is roughly one third of the overhead required by POSEIDON. Moreover, by comparing the complexity, we can infer that the homomorphic multiplication of the matrices in \Name is only linearly related to the dimension of the matrix, and is independent of the number of neurons in each layer of the model. On the contrary, the complexity of AP increases linearly with $\max_{i\in [\mathbb{L}]}\omega_i$. This implies that \Name  is more suitable for complex network architectures than POSEIDON.

\begin{figure*}[htb]
  \centering
  \subfigure[]{\label{4a}\label{Reprint_accuracy}\includegraphics[width=0.24\textwidth]{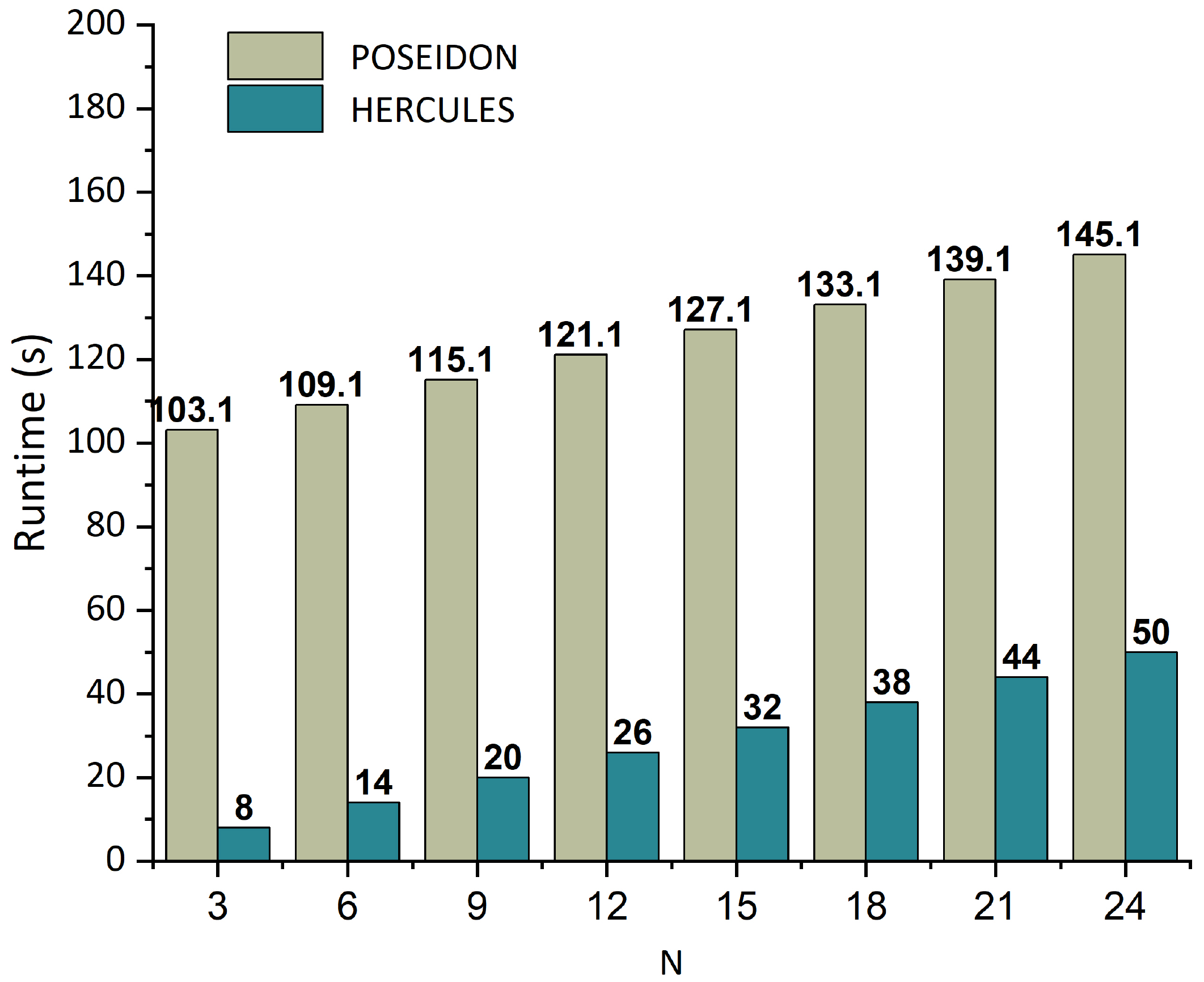}}
 \subfigure[]{\label{4b}\includegraphics[width=0.24\textwidth]{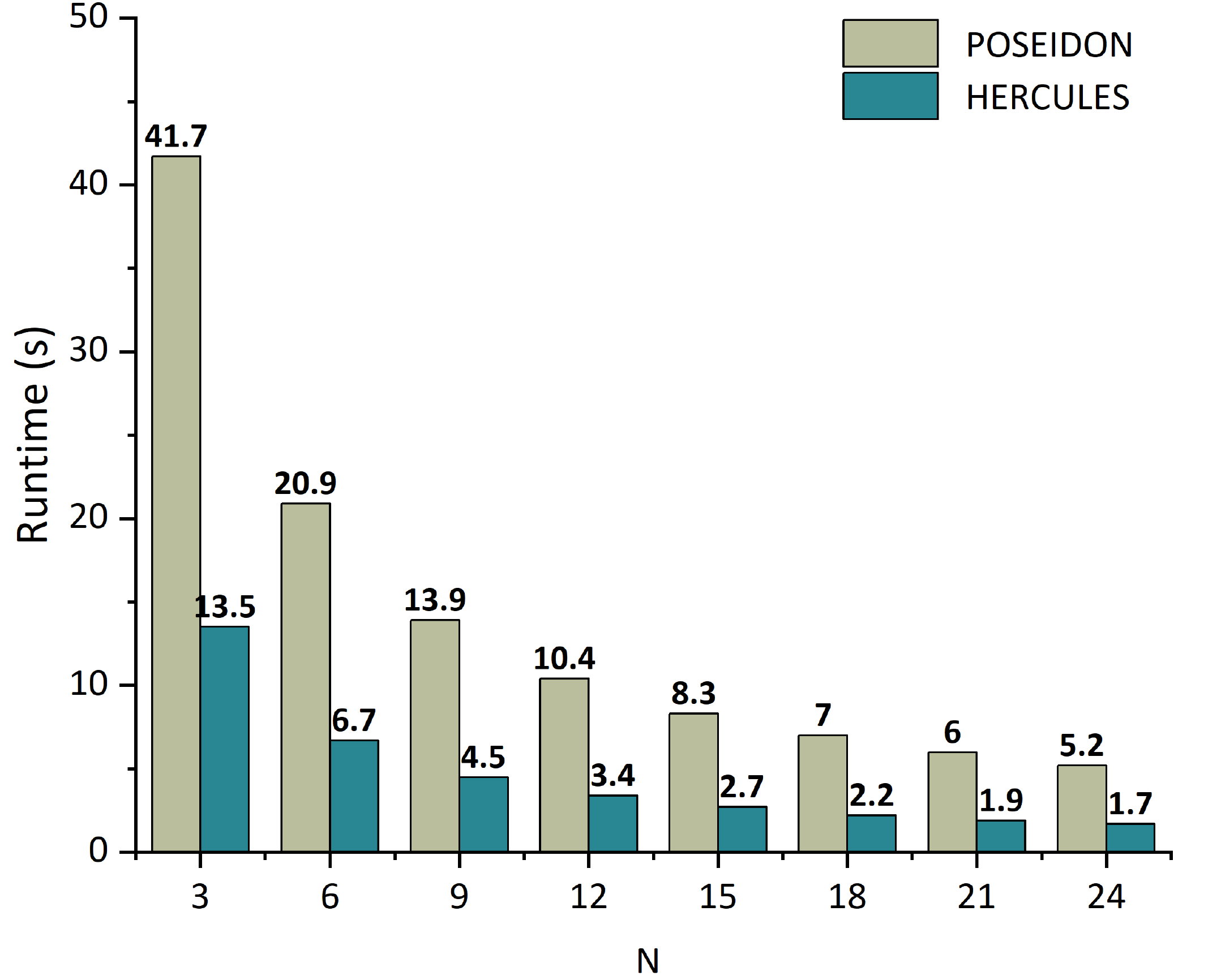}}
  \subfigure[]{\label{4c}\includegraphics[width=0.24\textwidth]{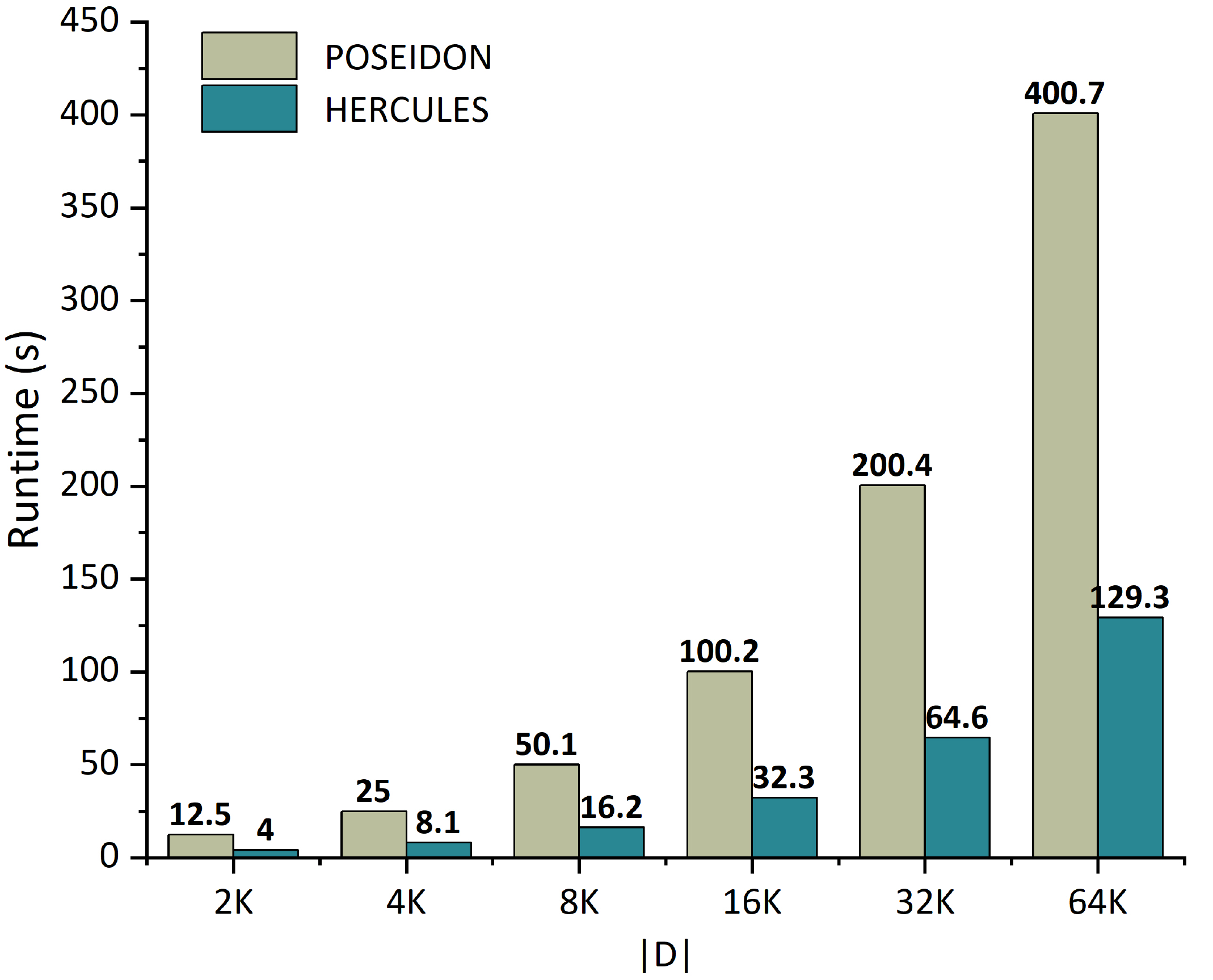}}
  \subfigure[]{\label{4d}\includegraphics[width=0.24\textwidth]{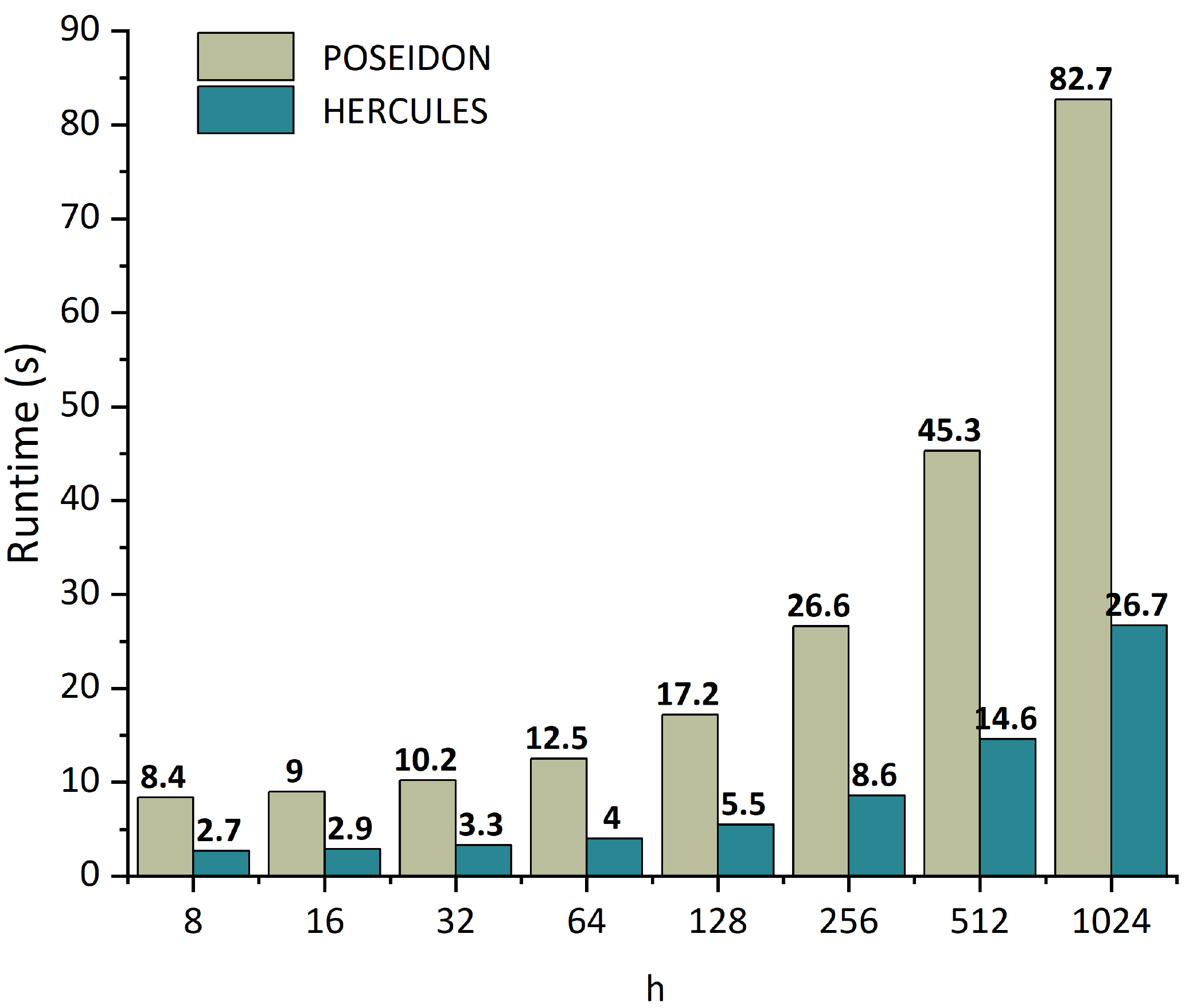}}
  \caption{Running time of one training epoch. (a) Increase the number of users $N$ when the number of samples for each user is $|D_i|=200$.  (b)  Increase the number of users $N$ when the total sample size is $|D|=2000$. (c) Increase the total sample size $|D|$ when  $N=10$. (d) Increase the sample dimension when $N=10$ and $|D|=200\times N$. }
  \label{runtime}
  \vspace{-15pt}
\end{figure*}

\begin{figure*}[htb]
  \centering
  \subfigure[]{\label{5a}\includegraphics[width=0.23\textwidth]{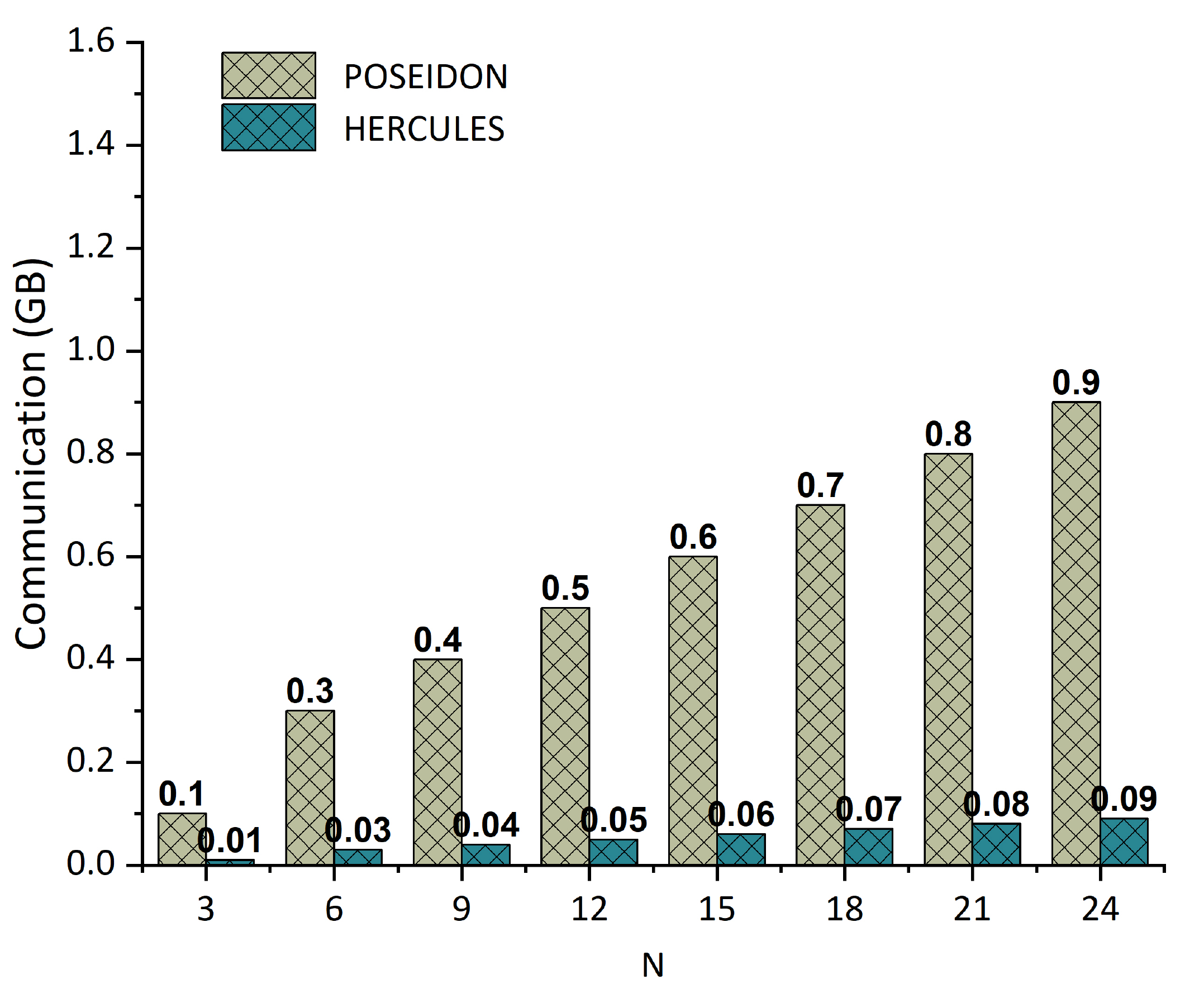}}
 \subfigure[]{\label{5b}\includegraphics[width=0.23\textwidth]{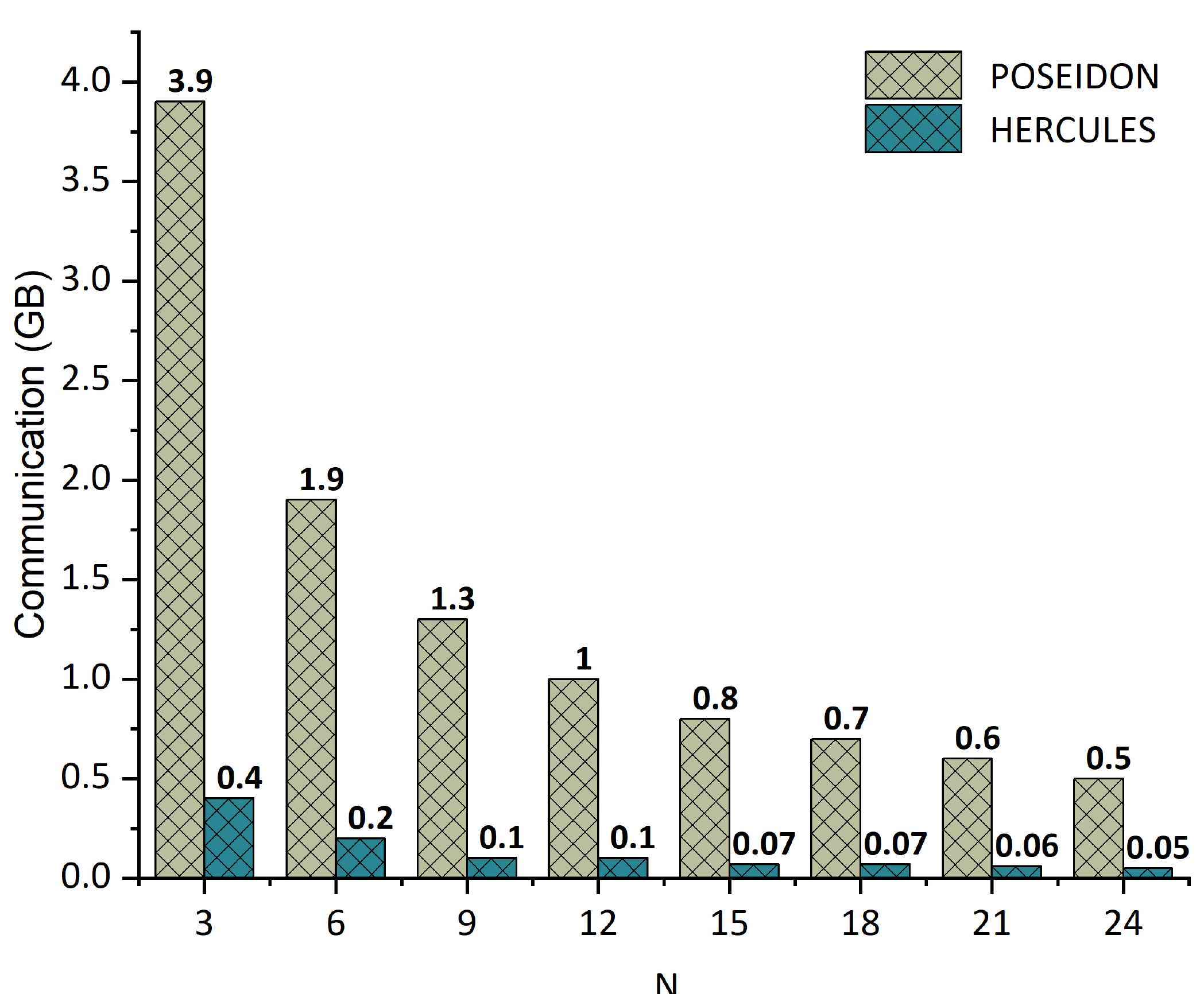}}
  \subfigure[]{\label{5c}\includegraphics[width=0.23\textwidth]{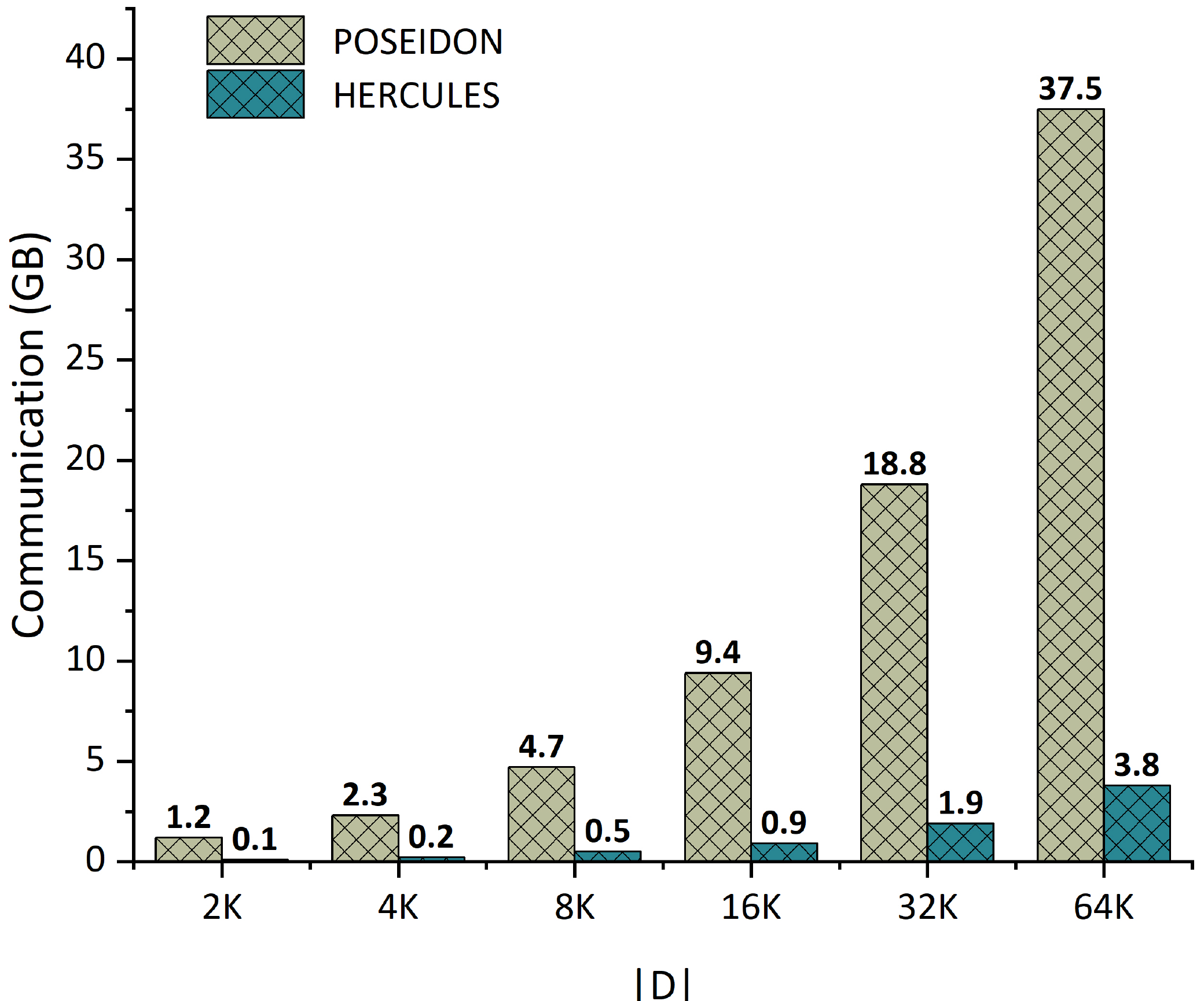}}
  \subfigure[]{\label{5d}\includegraphics[width=0.23\textwidth]{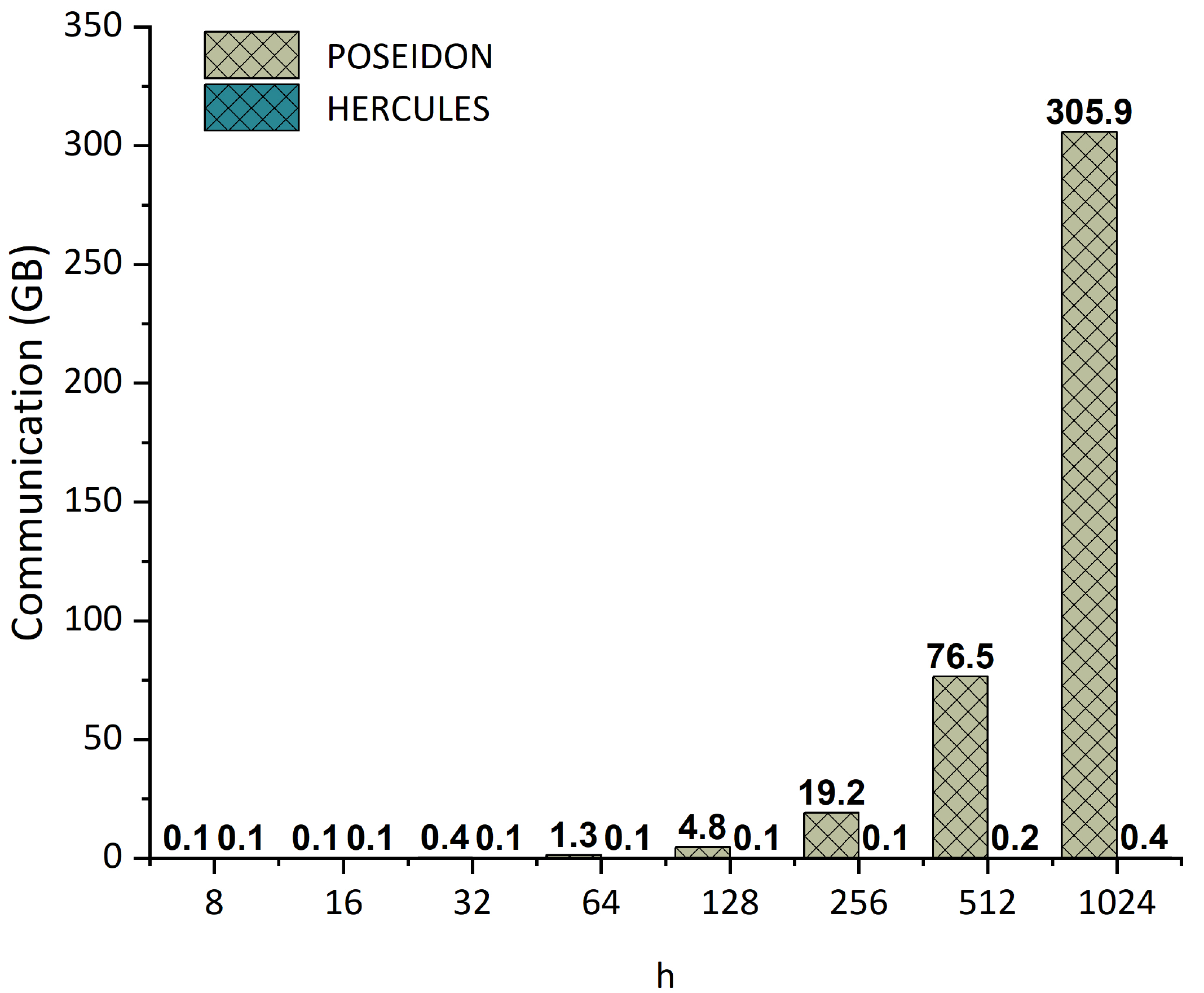}}
  \caption{Evaluation of communication overhead for one training epoch. (a) Increase the number of users $N$ when given  $|D_i|= 200$ for each user $i$.  (b)  Increase the number of users $N$ when the total sample size $|D|=2000$. (c) Increase the total sample size $|D|$ when  given  $N=10$. (d) Increase the dimension of a single sample when given  $N=10$ and $|D|=200\times N$. }
  \label{comm-overhead}
\end{figure*}

We further analyze the scalability of \Name and POSEIDON in terms of the number of users $N$, the number of samples $|D|$, and the number of dimensions $h$ for one sample. Here we use a two-layer architecture with 64 neurons in each layer. The local batch size for each user is 10. Figure~\ref{runtime} shows the experimental results, where we record the execution time of one training epoch, i.e., all the data of each user are processed once.  Specifically, Figure~\ref{4a} shows the execution time as the number of users grows, where we fix the number of samples held by each user as 200, and the dimension of each sample as 64. We can observe that the execution time of \Name and POSEIDON  shows a slight linear increase with the increase of the number of users. This stems from the fact that most of the operations performed by each user are concentrated locally except for the distributed bootstrapping procedure. Obviously, the percentage of $\mathbf{DBootstrap}$ operations over the total operations under ciphertext training is relatively small. We further fix the total number of samples in the system as 2000, and calculate the execution time of each user as the number of users increases. As shown in Figure~\ref{4b},  this causes a linear decrease in execution time since the increase in user data reduces the sample volume held by each user. Given the fixed number of users $N=10$ and $h=64$, Figure~\ref{4c} shows that the execution time of each user increases linearly as $|D|$ increases. It is obvious that the increase in $|D|$ implies an increase in the number of samples in each user. Figure~\ref{4d} also shows similar results under different sample dimensions.

 In general, \Name and POSEIDON show similar relationships in terms of computation cost under different hyper-parameters. However,
 we can observe that the running time of \Name is far less than that of POSEIDON, due to the superiority of our new matrix multiplication method.

\subsection{Communication Overhead}
Tables~\ref{Model accuracy1} and \ref{Model accuracy2} show the total communication overhead required by  \Name and POSEIDON over different datasets. We can observe that during the training process, the ciphertext data that each user needs to exchange with other parties in \Name are much smaller than that of POSEIDON. This also stems from the superiority of the new matrix multiplication method we design. Specifically, In  POSEIDON, AP is used for matrix multiplication to achieve fast SIMD operations. However, as shown in the fourth row of {Protocol 3} in \cite{sav2020poseidon}, this method requires multiple copies and zero padding operations for each row or column of the input matrix, depending on the number of neurons in each hidden layer, the absolute value of the difference between the row or column dimension of the matrix and the number of  neurons in the corresponding hidden layer. In fact, AP is an encoding method that trades redundancy in storage for computation acceleration. On the contrary, our method does not require additional element copy except for a small amount of zero padding in the initial stage to facilitate calculations. Therefore, \Name obviously exhibits smaller communication overhead. For example, given the MNIST dataset, a 3-layer fully connected model with 64 neurons per layer, the communication overhead of each user in  POSEIDON is about 703(GB) to complete 1000 global iterations, while \Name only needs 17.58(GB) per user.

We also analyze the scalability of \Name and POSEIDON in terms of the number of users $N$, the number of samples $|D|$, and the sample dimension $h$. Here we use a two-layer model architecture with 64 neurons in each layer. The local batch size for each user is 10. Figure~\ref{comm-overhead}  shows the experimental results. Similar to the results for computation cost comparison, we can observe that \Name exhibits better scalability compared to POSEIDON under  different hyper-parameters.  In addition, we also show the storage overhead advantage of \Name compared to POSEIDON, and discuss the performance of \Name  compared with other advanced MPC-based solutions.  More details can be found in Table~\ref{Storage overhead}  and Table~\ref{Comparison with existing MPC-based works} in Appendix.

\section{Conclusion}
\label{sec:conclusion}
In this paper, we propose \Name for privacy-preserving FL.  We  design a novel matrix coding technique to accelerate the training performance under ciphertext. Then, we use a novel approximation strategy  to improve the compatibility of \Name for processing non-polynomial functions. Experiments on  benchmark datasets demonstrate the superiority of \Name compared with  existing works. In the future, we will focus on designing more efficient optimization strategies to further reduce the computation overhead of \Name, to make it more suitable for practical applications.

\ifCLASSOPTIONcaptionsoff
\newpage \fi
\balance
\bibliographystyle{IEEEtran}
\bibliography{PPDR}
\clearpage
\appendices
\section*{Appendix}
\setcounter{section}{0}
\section{Security Extensions}
\label{Security Extensions}
\subsection{Active Adversaries} In  \Name, we consider a passive-adversary model with collusion of up to $N-1$ users. However, our work can also be extended to models in the presence of active adversaries. This can be achieved through verifiable computing techniques including zero-knowledge proofs (ZKF) \cite{weng2021mystique} and redundant computation \cite{zheng2019helen}. An intuitive idea is to use ZKF to ensure the format correctness of the ciphertext sent by the adversary, and use redundant calculations such as SPDZ \cite{keller2020mp} to verify the integrity of each party's calculation results. This would come at the cost of an increase in the computation complexity, that will be analyzed as future work.

\subsection{Out-of-the-Scope Attacks} The goal of \Name is to protect the privacy of users' local data and model parameters. However, there are still some attacks that can be launched from the prediction results of the model, e.g., membership inference \cite{shokri2017membership}, model inversion \cite{zhang2020secret}, or model stealing \cite{juuti2019prada}. These attacks can be mitigated by complementary countermeasures that are also easily integrated into \Name. For example, to defense the membership inference attack, on can add a carefully crafted noise vector to a confidence score vector to turn it into an adversarial example that misleads the attacker's classifier \cite{jia2019memguard}. We can also limit the number of prediction
queries for the queriers, thereby reducing the risk of adversaries launching model stealing attacks \cite{juuti2019prada}. We leave the combination of HE-based methods with existing defenses against various types of attacks as future works.

\section{The Baby-Step/giant-Step Algorithm Used in Algorithm~\ref{algorithm 3} }
\label{AP:the baby-step/giant-step algorithm}
 Given an integer $k\in (-h, h)$, we can rewrite $k=\sqrt{h} \cdot i +j$ where $-\sqrt{h}<i<\sqrt {h}$ and $0\leq j<\sqrt{h}$. Therefore,  Eq.(\ref{eq5}) can be parsed as

\begin{small}
 \begin{equation*}
\begin{split}
U^{\mu}\cdot \mathbf{a}&=\sum_{-\sqrt{h}<i<\sqrt{h}, 0\leq j<\sqrt{h}}(\mathbf{u}_{\sqrt{h}\cdot i+j}^{\mu}\odot\mathbf{R}(\mathbf{a},\sqrt{h}\cdot i+j))\\
&=\sum_{-\sqrt{h}<i<\sqrt{h}}\mathbf{R}(\sum_{0\leq j< \sqrt{h}}\mathbf{a}_{i,j}, \sqrt{h}\cdot i)
\end{split}
\end{equation*}
\end{small}
where $\mathbf{a}_{i,j}=\mathbf{R}(\mathbf{u}^{\mu}_{\sqrt{h}\cdot i+j}, -\sqrt{h}\cdot i)\odot \mathbf{R}(\mathbf{a}, j)$. Based on this, we can first compute $\mathbf{R}(\mathbf{a}, j)$ for $0\leq j< \sqrt{h}$, and then use them to compute the encryption of $\mathbf{a}_{i,j}$. In general, Step 1-1 requires homomorphic operations of $2h$ additions, $2h$ constant multiplications, and $3\sqrt{h}$ rotations. Similarly, Step 1-2 can be completed by  $h$ additions, $h$ constant multiplications, and $2\sqrt{h}$ rotations.

The number of constant multiplications required in Step 2 can also be reduced by exploiting two-input multiplexers. We can observe that

\begin{small}
 \begin{equation*}
\begin{split}
&\mathbf{Mul}_{pt}(\mathbf{Rot}([\mathbf{A}^{(0)}]_{pk}, k-h), \mathbf{v}_{k-h})\\
&=\mathbf{Rot}(\mathbf{Mul}_{pt}([\mathbf{A}^{(0)}]_{pk}, \mathbf{R}(\mathbf{v}_{k-h}, h-k)), k-h)\\
&=\mathbf{Rot}([\mathbf{A}^{(0)}]_{pk}-\mathbf{Mul}_{pt}([\mathbf{A}^{(0)}]_{pk}, \mathbf{R}(\mathbf{v}_{k}, -k)), k-h)
\end{split}
\end{equation*}
\end{small}

Then, we can first  compute $\mathbf{Mul}_{pt}([\mathbf{A}^{(0)}]_{pk}, \mathbf{R}(\mathbf{v}_{k}, -k))$ for each $1\leq k<h$. Based on the fact that

 \begin{small}
 \begin{equation*}
\begin{split}
\mathbf{Mul}_{pt}(\mathbf{Rot}([\mathbf{A}^{(0)}]_{pk}, k), \mathbf{v}_{k})=\mathbf{Rot}(\mathbf{Mul}_{pt}([\mathbf{A}^{(0)}]_{pk},  \mathbf{R}(\mathbf{v}_{k}, -k)), k)
\end{split}
\end{equation*}
\end{small}
we can get the ciphertext  $[\mathbf{A}^{(k)}]_{pk}$ with addition and rotation operations.

\section{Matrix Transposition on Packed Ciphertexts}
\label{APP:Matrix Transposition on Packed Ciphertexts}
In this section, we introduce how to homomorphically transpose a matrix under the packed ciphertext. Let $\mathbf{a}=\iota^{-1}(A)\in R^{n}$ indicate the vector representation of a matrix $A$, $U^T$ is defined as the matrix representation of the transpose map $A\mapsto A^T$ on $R^{h\times h} \cong R^{n}$.  Then, for $0\leq i, j< h$, each element in $U^T$ can be expressed as

\begin{small}
$$ {U}_{h\cdot i+j, k}^{T}=\left\{
\begin{aligned}
&1 \quad \mathsf{if}\; k=h\cdot j +i;\\
& 0 \quad \mathsf{otherwise}.\\
\end{aligned}
\right.
$$
\end{small}
 Hence, the $k$-th diagonal vector of $U^T$ is nonzero if and only if $k=(h-1)\cdot i$ for some $i\in \mathbb{Z} \cap (-h, h)$. As a result, $U^T$ has a total of $(2h-1)$ nonzero diagonal vectors as a sparse matrix. Therefore, linear transformation for matrix transpose can be expressed as

 \begin{small}
 \begin{equation}
\begin{split}
U^{T}\cdot \mathbf{a}=\sum_{-h<i<h}(\mathbf{t}_{(h-1)\cdot i}\odot\mathbf{R}(\mathbf{a},(h-1)\cdot i))
\end{split}
\end{equation}
\end{small}
where we use $\mathbf{t}_{(h-1)\cdot i}$ to represent the nonzero diagonal vector of $U^T$. When $i\geq 0$, the $t$-th component of the vector  $\mathbf{t}_{(h-1)\cdot i}$ is computed by

\begin{small}
$$ \mathbf{t}_{(h-1)\cdot i}[t]=\left\{
\begin{aligned}
&1 \quad \mathsf{if}\; t-i= (h+1)\cdot j, 0\leq j< h-i;\\
& 0 \quad \mathsf{otherwise}.\\
\end{aligned}
\right.
$$
\end{small}
If $i\leq 0$, we have
\begin{small}
$$ \mathbf{t}_{(h-1)\cdot i}[t]=\left\{
\begin{aligned}
&1 \quad \mathsf{if}\; t+i= (h+1)\cdot j, 0\leq j< h+i;\\
& 0 \quad \mathsf{otherwise}.\\
\end{aligned}
\right.
$$
\end{small}
In general, the computation cost required for matrix transposition is about $2h$ rotations, which can be further reduced to $3\sqrt{h}$ by the previously described baby-step/giant-step method.

\section{Rectangular Matrix Multiplication on Packed Ciphertexts}
\label{APP:Rectangular Matrix Multiplication on Packed Ciphertexts}
We extend the homomorphic multiplication between square matrices to general rectangular matrices, such as $R^{t\times h}\times R^{h\times h}\rightarrow R^{t\times h}$ or $R^{h\times h}\times R^{h\times t}\rightarrow R^{h\times t}$. Without loss of generality, we consider $t<h$.  For a $(h_1\times h)$-dimensional matrix $A_1$ and a $(h_2\times h)$-dimensional matrix $A_2$, we use $(A_1; A_2)$ to represent the $(h_1+h_2)\times h$  matrix which is obtained by concatenating $A_1$ and $A_2$ in a vertical direction. Similarly, $(A_1|A_2)$ indicates the matrix concatenated in a horizontal direction if matrices $A_1$ and $A_2$ have the same number of rows.

A naive solution to perform multiplication between rectangular matrices is to convert any matrix into a square matrix through zero padding, and then use \textbf{Algorithm}~\ref{algorithm 3} to implement the multiplication between square matrices homomorphically. This leads to a rotation with a complexity of $O(h)$. We provide an improved method through insight into the property of matrix permutation.\\
\circled{1}\textbf{Refined Rectangular Matrix Multiplication}. Given a $t\times h$ matrix $A$ and an $h\times h$ matrix $B$ with $t$ divide $h$, we introduce a new symbol $C_{t_1: t_2}$,  which is denoted  the $(t_2-t_1)\times h$ submatrix of $C$ formed by extracting from $t_1$-th row to the $(t_2-1)$-th row of $C$.  As a result, the product of $AB$ can be expressed as below.

\begin{small}
 \begin{equation}
\begin{split}
A\cdot B&=\sum_{k=0}^{h-1}(\phi^{k}\circ \mu(A) )\odot\left((\pi^{k}\circ \zeta (B))_{0: t}\right)\\
        &=\sum_{0\leq i< t}\sum_{0\leq j<h/t} (\phi^{j\cdot t+i}\circ \mu(A))\odot\left((\pi^{j\cdot t+i}\circ \zeta (B))_{0: t}\right)
\end{split}
\end{equation}
\end{small}
Our key observation is the following lemma, which provides us with ideas for designing fast rectangular matrix multiplication.
\begin{lemma}
\label{lemma1}
Permutation $\phi$ and $\mu$ are commutative.  For $k>0$, we have  $\phi^{k} \circ \mu= \mu \circ\phi^{k}$. Also,   $\pi^{k}\circ \zeta =\zeta\circ \pi^{k}$ for $k>0$.
\end{lemma}

Based on  Lemma~\ref{lemma1}, we define  an $h\times h$-dimensional matrix $\tilde {A}$, which contains $h/t$ copies of $A$ from the vertical direction (i.e., $\tilde {A}=(A; \cdots ; A)$). It means that

\begin{small}
 \begin{equation}
\begin{split}
(\phi^{i}\circ \mu(\tilde{A}))_{j\cdot t: (j+1)\cdot t}&=(\phi^{i}\circ \mu(\tilde{A})_{j\cdot t: (j+1)\cdot t})\\
                                                       &=\phi^{i}\circ \mu \circ \phi^{j\cdot t}(A)\\
                                                       &=\phi^{j\cdot t+i}\circ \mu(A)
\end{split}
\end{equation}
\end{small}
 Then, we can further compute
\begin{small}
 \begin{equation}
\begin{split}
(\pi^{i}\circ \zeta({B}))_{j\cdot t: (j+1)\cdot t}=(\pi^{j\cdot t+i}\circ \zeta(B))_{0:t}
\end{split}
\end{equation}
\end{small}
As a result, the matrix product $AB$ can be rewritten as below.
\begin{small}
 \begin{equation*}
\begin{split}
A\cdot B=\sum_{0\leq j< h/t}\left(\sum_{0\leq i<t} (\phi^{i}\circ \mu(\tilde{A}))\odot(\pi^{i}\circ \zeta({B}))\right)_{j\cdot t: (j+1)\cdot t}
\end{split}
\end{equation*}
\end{small}\\
\circled{2}\textbf{Homomorphic Rectangular Matrix Multiplication}. Given two ciphertexts $[\mathbf{\tilde{A}}]_{pk}$ and $[\mathbf{{B}}]_{pk}$,  we can first compute $\mu{(\tilde{A})}$ and  $\zeta({B})$ utilizing the baay-step/giant-step approach. Then, $\sum_{0\leq i<t} (\phi^{i}\circ \mu(\tilde{A}))\odot(\pi^{i}\circ \zeta({B}))$ can be securely computed  in a similar way to \textbf{Algorithm} \ref{algorithm 3}.
Finally, we get the encryption of ${A} B$  by performing aggregation and rotations  operations. The details are shown in  \textbf{Algorithm}~\ref{Complexity of Algorithm 41}.

 \begin{algorithm}
\caption{ \small{Homomorphic Rectangular Matrix  Multiplication} }
\label{algorithm 41}
\begin{algorithmic}[1]
\footnotesize
\REQUIRE \texttt{HE-RMatMult} $([\mathbf{\tilde{A}}]_{pk}, [\mathbf{B}]_{pk})$
\STATE  $[\mathbf{A}^{(0)}]_{pk}\leftarrow $ \texttt{HE-LinTrans} $([\mathbf{\tilde{A}}]_{pk}, U^{\mu})$
\STATE $[\mathbf{B}^{(0)}]_{pk}\leftarrow $ \texttt{HE-LinTrans} $([\mathbf{A}]_{pk}, U^{\zeta})$
\FOR {$k=1$ to $t-1$}
\STATE   $[\mathbf{A}^{(k)}]_{pk}\leftarrow $ \texttt{HE-LinTrans} $([\mathbf{A}^{(0)}]_{pk}, V^{k})$
\STATE   $[\mathbf{B}^{(k)}]_{pk}\leftarrow $ \texttt{HE-LinTrans} $([\mathbf{B}^{(0)}]_{pk}, P^{k})$
\ENDFOR
\STATE   $[\mathbf{\tilde{A}B}]_{pk} \leftarrow \mathbf{Mul}_{ct}([\mathbf{A}^{(0)}]_{pk}, [\mathbf{B}^{(0)}]_{pk})$
\FOR {$k=1$ to $t-1$}
\STATE $[\mathbf{\tilde{A}B}]_{pk}\leftarrow \mathbf{Add} ([\mathbf{\tilde{A}B}]_{pk}, \mathbf{Mul}_{ct}([\mathbf{A}^{(k)}]_{pk}, [\mathbf{B}^{(k)}]_{pk}))$
\ENDFOR
\STATE  $[\mathbf{AB}]_{pk}\leftarrow [\mathbf{\tilde{A}B}]_{pk}$
\FOR {$k=0$ to $\log(h/t)-1$}
\STATE  $[\mathbf{AB}]_{pk}\leftarrow \mathbf{Add} ([\mathbf{{A}B}]_{pk}, \mathbf{Rot}([\mathbf{{A}B}]_{pk}, t\cdot h\cdot 2^k))$
\ENDFOR
\STATE  \small{\textbf{return}}  $[\mathbf{AB}]_{pk}$
\end{algorithmic}
\end{algorithm}

\begin{table}[H]
\centering
\small
\caption{Complexity of Algorithm 6}
\label{Complexity of Algorithm 41}
\begin{tabular}{|c|c|c|c|c|}
\hline
\textbf{Step}&$Add$& $mul_{pt}$ &$Rot$ & $mul_{ct}$\\
\hline
\textbf{1}&$3h$&$3h$&$5\sqrt{h}$&-\\
\hline
\textbf{2}&$t$&$2t$&$3t$&-\\
\hline
\textbf{3}&$t$&-&-&$t$\\
\hline
\textbf{4}&$\log(h/t)$&-&$\log(h/t)$&-\\
\hline
\textbf{Total}&$3h+2t$&$4h$&$3h+5\sqrt{h}$&$t$\\
\hline
\end{tabular}
\end{table}
\textbf{Table}~\ref{Complexity of Algorithm 41} presents the total complexity of \textbf{Algorithm}~\ref{algorithm 41}. We observe that compared with \textbf{Algorithm}~\ref{algorithm 3}, it reduces the complexity of the rotation operations of steps 2 and 3 to $O(t)$, although additional  operation is required for step 4. Moreover, the final result $[\mathbf{AB}]_{pk}$ is a ciphertext of a $h\times h$-dimensional matrix containing $(h/t)$ copies of the expected matrix product $AB$ in a vertical direction. Hence, the resulting ciphertext remains the format as a  rectangular matrix, and it can be further performed matrix calculations without additional overhead.

\section{ Parallel  Matrix Computation}
\label{APP:Parallel  Matrix Computation}
The above matrix operations are all performed in the message space $R^{n}$, where we assume that $n=h^2$.  Actually, most HE schemes have a lot of plaintext slots (up to thousands) compared to the dimension of the matrix in deep learning, that is, usually $n\gg h^2$. Hence, most plaintext slots will be wasted if a ciphertext is only used to store a matrix. We provide a method of encrypting multiple matrices into a ciphertext, so as to realize parallel matrix calculation in a SIMD manner. Specifically, we assume that $n$ is divisible by $h$ and let  $\beta=n/h^2$. Then the encoding map described in above for singe matrix  will be modified as  $\iota_\beta: R^{n}\rightarrow (R^{h\times h})^{\beta}$. For an input vector $\mathbf{a}=(a_t)_{0\leq t<n}$, $\iota_\beta$ is defined as below.

\begin{small}
 \begin{equation}
\begin{split}
\iota_\beta : \mathbf{a}\mapsto \left(A_k=(a_{\beta (h\cdot i+j)+k})\right)_{0\leq k<\beta}.
\end{split}
\end{equation}
\end{small}
The components of $\mathbf{a}$ with indexes congruent to $k$ modulo $\beta$ are corresponding to the $k$-th matrix $A_k$.
We observe that for an integer $0\leq t<h^2$, the rotation operation $\mathbf{R}(\mathbf{a}, \beta t)$ represents the matrix-wise rotation  by $\beta$ positions. This can be naturally extended to other operations including scalar linear conversion and matrix multiplication. For example, when a single ciphertext is embedded in $\beta$ ($h\times h$)-dimensional matrices, we can immediately perform matrix multiplication between  $\beta$ pairs matrices by using the previously described algorithm on two ciphertexts. This total computation complexity is consistent with \textbf{Algorithm}~\ref{algorithm 3}, but results in a less amortized computation complexity of $O(h/\beta)$ for each matrix.

\section{Proof of Lemma ~\ref{lemma2}}
\label{AP:pof LE3}
\begin{proof}
We use induction to prove it. First, it is true for n=1.  Assume that $p_d=\frac{2d+1}{4^d}\begin{pmatrix}  2d \\  d \end{pmatrix}$ for some $d\leq 1$, Hence, we have

 \begin{small}
 \begin{equation*}
 \begin{split}
p_{d+1}&=p_d+\frac{1}{4^{d+1}}\begin{pmatrix}  2d+2 \\  d+1 \end{pmatrix}\\
       & =\frac{1}{4^{d+1}}\left(\frac{2(2d+2)!}{(d+1)!d!}+\frac{(2d+2)!}{(d+1)!(d+1)!}\right)\\
       &=\frac{2d+3}{4^{d+1}}\begin{pmatrix}  2d+2 \\  d+1 \end{pmatrix}
\end{split}
\end{equation*}
\end{small}

Hence, the lemma is proved by induction.
\end{proof}

\section{Proof of Lemma ~\ref{lemma3}}
\label{AP:pof LE4}

\begin{proof}
Obviously,  $g_d(m)\leq g_d(1)=1$ for $m\in[0, 1]$.  We define $G(m)=(1-m)^{p_d}-(1-g_d(m))$, and then we prove that $G(m)\geq 0$ for $m\in [0, 1]$ by showing \\
1. $G(0)=G(1)=0$.\\
2. there exists $m_0\in (0, 1)$ s.t. $G(m_0)>0$.\\
3. there exists a unique $y_0\in (0, 1)$ s.t. $G'(y_0)=0$.

We first explain that the above three conditions are derived from $G(m)\geq 0$.  Specifically, if there is a point $m\in (0, 1)$ such that $G(m_1)<0$, then according to the continuity of $G$, there is a root $m_2$ of the function $G$ between $m_0$ and $m_1$. From the mean value theorem, it is obvious that there exist $y_1\in (0, m_2)$ and $y_2\in (x_2, 1)$ satisfying $G'(y_1)=G'(y_2)=0$, which contradicts the third condition. We begin to prove these three conditions. The first one is trivial. For the second condition, we observe that $G(0)=0$, $G'(0)=0$ and $G''(0)>0$. We can infer that $G'(m)>0$ for $m\in (0, \epsilon)$ for some $\epsilon>0$, based on the continuity of $G''$. Further, since $G(0)=0$, we have $G(m)>0$ for $m\in (0, \epsilon)$ which implies the second condition.

To prove uniqueness, let $G'(m)=p_d(1-m^2)^d-p_d(1-m)^{p_d-1}=0$. Then we have $(1-m)^{d-p_d+1}\cdot (1+m)^d=1$ for $m\in (0, 1)$. Logarithmically, it holds that

 \begin{small}
 \begin{equation*}
 \begin{split}
\frac{\log(1+m)}{\log (1-m)}=-\frac{d-p_d+1}{d}
\end{split}
\end{equation*}
\end{small}
Since $\frac{\log(1+m)}{\log (1-m)}$ is a strictly increasing function, this is the unique $y_0\in (0,1)$ that satisfies $G'(y_0)=0$.
\end{proof}
\section{Proof of Lemma ~\ref{lemma4}}
\label{AP:pof LE5}

\begin{proof}
Let $y=1-m$, and define
\begin{small}
 \begin{equation*}
 \begin{split}
H(y)=\frac{p_d\cdot 2^d}{d+1}\cdot y^{d+1}-(1-g_d(1-y))
\end{split}
\end{equation*}
\end{small}
Then $H(y)'=p_d\cdot 2^d \cdot y^d -g_d'(1-y)=p_d\cdot 2^d\cdot y^d-p_d\cdot y^d(2-y)^d\geq 0$ for $y\in [0, 1]$.  Based on $H(0)=0$, it holds that $H(y)\geq 0$. Hence, for $m\in [0, 1]$, we have
\begin{small}
 \begin{equation*}
 \begin{split}
1-g_d(m)\leq \frac{p_d\cdot 2^n}{d+1}\cdot (1-m)^{d+1}\leq 2^d\cdot (1-m)^{d+1}
\end{split}
\end{equation*}
\end{small}
\end{proof}

\section{Experimental comparisons with work \cite{cheon2019numerical}}
\label{Experimental comparisons}
\begin{table}
\centering
\small
\caption{ Comparisons with work \cite{cheon2019numerical}}
\label{comparisons with work116}
\begin{tabular}{|c|c|c|}
\Xhline{1pt}
$\sigma$& Work \cite{cheon2019numerical} &\textbf{Ours}\\ \hline
$8$& $208$ s ($50.78$ ms) & $21$ s ($5.12$ ms)\\ \hline
$10$& $307$ s ($74.9$ ms) & $27$ s ($6.59$ ms)\\ \hline
$12$& $532$ s ($129.8$ ms) & $37$ s ($9.03$ ms)\\ \hline
$14$& $823$ s ($200.9$ ms) & $51$ s ($12.45$ ms)\\ \hline
$16$& $1392$ s ($339.8$ ms) & $70$ s ($17.08$ ms)\\ \hline
$18$& $1930$ s ($471.11$ ms) & $76$ s ($18.5$ ms)\\ \hline
$20$& $2740$ s ($668.9$ ms) &$84$ s ($20.05$ ms)\\ \Xhline{1pt}
\end{tabular}
\end{table}
TABLE~\ref{comparisons with work116} shows the running times of  work \cite{cheon2019numerical} and our method for approximation of the sign function, where $\delta$  is used to represent the error bound, i.e.,  $2^{-\sigma}$ means that the value error of the original function and the approximated function is  small than  $2^{-\sigma}$, $\mathcal{N}=2^{14}$ and  $Q_\mathcal{L}=2^{2250}$. We observe that the running time of  our method is much smaller than that of \cite{cheon2019numerical}. For example, the amortized running time  of \cite{cheon2019numerical} to obtain the approximate result of sign function within $2^{-20}$ errors  is 668.9 milliseconds (amortized running time),  while the running time of our method is only 20.05 milliseconds,  which is $33$ times faster than \cite{cheon2019numerical}.

\section{Details Implementation of Hercules}
\label{AP:Details Implementation of Hercules}
Figure~\ref{Detailed description of the scheme}  presents the details of implementing \Name, which essentially executes \textbf{Algorithm}~\ref{algorithm 1} under the ciphertext.   Note that we assume that all matrices involved in training have the same dimensions for simplicity. In actual implementation, the user can adaptively select the dimension so as to use \textbf{Algorithm} ~\ref{algorithm 3} or \ref{algorithm 41} to obtain the ciphertext result. $\mathbf{RS}(\cdot)$ is used to resulting ciphertext after each multiplication. This means that for the ciphertext with an initial layer number of $\mathcal{L}$, the maximum depth of $\mathcal{L}$ ciphertext multiplication can be evaluated.  In practical applications, assuming $Q_\mathcal{L}/\Delta=r$, the ciphertext can be rescaled after $r$ multiplications instead of after each multiplication. On the other hand, Figure~\ref{Detailed description of the scheme} only requires each user to perform the distributed bootstrapping function (i.e.,  $\mathbf{DBootstrap}$)  when calculating the activation function, because this process usually requires more circuit depth. Similarly, in practice, users can adaptively  execute   $\mathbf{DBootstrap}$  as long as the depth of the circuit is about to be exhausted.

\renewcommand\tablename{Fig.}
\renewcommand \thetable{\arabic{table}}
\setcounter{table}{4}

\begin{table*}[!htb]
\small
\centering
\begin{tabular}{|p{17.5cm}|}
\hline\\
\qquad \qquad \qquad \qquad \qquad \qquad \qquad \qquad \qquad \qquad \quad \qquad {Implementation of Hercules} \\
\begin{itemize}
\item \textbf{Prepare}:
  \begin{itemize}
    \item[-] This process is exactly the same as described in \textbf{Algorithm}~\ref{algorithm 6}.
\end{itemize}
\item[] For $k=0$ to $k=H-1$, $\mathcal{C}$ and all users perform the following operations in concert.
  \item \textbf{Feedforward}:
  \begin{itemize}
    \item[-] Each user ${P}_i$ computes  $\hat{X_i}$, $\hat{Y_i}$ by the function $\mathbf{Ecd}(\cdot)$ and encrypts them as $[\mathbf{X_i}]_{pk}$ and $[\mathbf{Y_i}]_{pk}$.
    \item[] For $j=1$ to $j=\mathbb{L}$, each user ${P}_i$ performs following operations.
  \begin{itemize}
     \item [-]  Compute $[\mathbf{E_{j, i}^{k}}]_{pk}=$\texttt{HE-MatMult} $(\mathbf{[\omega_{j}^{k}}]_{pk}, [\mathbf{M_{j-1, i}^{k}}]_{pk})$, where  $E_{j, i}^{k}={\omega}_{j}^{k}\times M_{j-1, i}^{k}$, $\mathbf{RS}([\mathbf{E_{j, i}^{k}}]_{pk})$.
    \item[-]Compute $[\mathbf{M_{j, i}^{k}}]_{pk}$, where $M_{j, i}^{k}={\varphi(E_{j, i}^{k})}$, $\mathbf{RS}([\mathbf{M_{j, i}^{k}}]_{pk})$.  $\varphi(\cdot)$ is approximated as a polynomial in advance.
    \item[-] All users collaboratively compute $\mathbf{DBootstrap}([\mathbf{M_{j, i}^{k}}]_{pk}, \mathcal{L}_{[\mathbf{M_{j, i}^{k}}]_{pk}}, \Delta_{[\mathbf{M_{j, i}^{k}}]_{pk}}, \{sk_i\})$.
  \end{itemize}
  \end{itemize}
\item \textbf{Backpropagation}:
 \begin{itemize}
    \item[-] Each $P_i$ computes  $[\mathbf{L_{\mathbb{L}, i}^{k}}']_{pk}=\mathbf{Sub}([\mathbf{Y_i}]_{pk}, \mathbf{ [M_{\mathbb{L},i}^{k}}]_{pk})$, $[\mathbf{L_{\mathbb{L}, i}^{k}}]_{pk}=\mathbf{Mul}_{ct}([L_{\mathbb{L}, i}^{k}]_{pk}, [L_{\mathbb{L}, i}^{k}]_{pk})$, where $L_\mathbb{L}^{k}=||y[t]-M_\mathbb{L}^{k}||_{2}$.
    \item[-]Each $P_i$ computes $[\mathbf{L_{\mathbb{L},i}^{k}}]_{pk}$ with basic SIMD operations, where $L_\mathbb{L}^{k}=\varphi'(E_\mathbb{L}^{k})\odot L_\mathbb{L}^{k}$.  Then, computes $\mathbf{RS}([\mathbf{L_{\mathbb{L},i}^{k}}]_{pk})$.
    \item[-]  Each $P_i$ computes $[\mathbf{(M_{\mathbb{L}-1, i}^{k})}^{T}]_{pk}= \texttt{HE-LinTrans} ({\mathbf{[M_{\mathbb{L}-1, i}^{k}]}_{pk}}, U^{(\mathbb{L}-1)})$, where $U^{(\mathbb{L}-1)}$  is the  permutation representation matrix of ${(\mathbf{M_{\mathbb{L}-1, i}^{k})}}^{T}$.
    \item[-] Each $P_i$ computes $[\bigtriangledown \mathbf{{\omega}_{\mathbb{L}, i}^{k}}']_{pk}=$\texttt{HE-MatMult} $([\mathbf{(M_{\mathbb{L}-1, i}^{k})}^{T}]_{pk}, [\mathbf{{L}_{\mathbb{L}, i}^{k}}]_{pk})$. Then, computes $\mathbf{RS}([\bigtriangledown \mathbf{{\omega}_{\mathbb{L}, i}^{k}}']_{pk})$.
    \item[-]  Each $P_i$ computes $[\bigtriangledown \mathbf{{\omega}_{\mathbb{L}, i}^{k}}]_{pk}=\mathbf{Add}([\bigtriangledown \mathbf{{\omega}_{\mathbb{L}, i}^{k}}']_{pk}, \mathbf{[\bigtriangledown \mathbf{{\omega}_{\mathbb{L}, i}^{k}}]}_{pk})$, and $\mathbf{RS}([\bigtriangledown \mathbf{{\omega}_{\mathbb{L}, i}^{k}}]_{pk})$, where $\bigtriangledown{\omega}_{\mathbb{L}, i}^{k}=\bigtriangledown{{\omega}_{\mathbb{L}, i}^{k}}+ {(M_{\mathbb{L}-1, i}^{k})}^{T}\times {L_{\mathbb{L}, i}^{k}}$.
    \item[] For $j=\mathbb{L}-1$ to $j=1$, each user ${P}_i$ performs following operations.
    \begin{itemize}
    \item[-] Compute $[\mathbf{\omega_{j+1, i}^{k})}^{T}]_{pk}= \texttt{HE-LinTrans} ([\mathbf{\omega_{j+1, i}^{k})}]_{pk}, U^{(\mathbb{L}-1)})$, where $U^{(j+1)})$  is the  permutation representation matrix of $(\omega_{j+1, i}^{k})^{T}$.
     \item[-]  Compute $[\mathbf{L_j^{k}}]_{pk}=$\texttt{HE-MatMult} $([\mathbf{L_{j+1}^{k}}]_{pk}, [\mathbf{\omega_{j+1, i}^{k}}^{T}]_{pk})$.Then, computes $\mathbf{RS}([\mathbf{L_j^{k}}]_{pk})$, where $L_j^{k}=L_{j+1}^{k}\times (\omega_{j+1}^{k})^{T}$.
    \item[-] Compute $[\mathbf{L_{j, i}^{k}}]_{pk}$ with basic SIMD operations, where $L_{j, i}^{k}=\varphi'(E_{l, i} ^{k})\odot L_{j, i}^{k}$.  Then, computes $\mathbf{RS}([\mathbf{L_{j, i}^{k}}]_{pk})$.
      \item[-] Compute $[\bigtriangledown \mathbf{{\omega}_{j, i}^{k}}']_{pk}=$\texttt{HE-MatMult} $([\mathbf{(M_{j-1, i}^{k})}^{T}]_{pk}, [\mathbf{L}_{j, i}^{k}]_{pk})$. Then, computes $\mathbf{RS}([\bigtriangledown \mathbf{{\omega}_{j, i}^{k}}']_{pk})$.
    \item[-]  Compute $[\bigtriangledown \mathbf{{\omega}_{j, i}^{k}}]_{pk}=\mathbf{Add}([\bigtriangledown \mathbf{{\omega}_{j, i}^{k}}']_{pk}, \mathbf{[\bigtriangledown \mathbf{{\omega}_{j, i}^{k}}]}_{pk})$, and $\mathbf{RS}([\bigtriangledown \mathbf{{\omega}_{j, i}^{k}}]_{pk})$, where $\bigtriangledown{\omega}_{j, i}^{k}=\bigtriangledown{{\omega}_{j, i}^{k}}+ {(M_{j-1, i}^{k})}^{T}\times {L_{j, i}^{k}}$.
  \end{itemize}
  \end{itemize}
 \item \textbf{Aggregation}:
  \begin{itemize}
    \item[] For $j=1$ to $j=\mathbb{L}$,  the cloud server $\mathcal{C}$ performs following operations.
  \begin{itemize}
     \item [-] Compute $[\mathbf{{\bigtriangledown \omega}_{j}^{k}}]_{pk}=[\sum_{i=1}^{N}\mathbf{\bigtriangledown {\omega}_{j, i}^{k}}]_{pk}$ with the basic $\mathbf{Add}$ function.
  \item[-] Compute $[\mathbf{{\omega}_{j}^{k+1}}']_{pk}=\mathbf{Mul}_{pt}([\mathbf{\bigtriangledown {\omega}_{j}^{k}}]_{pk}, \frac{\eta}{\mathcal{B}\times N})$.
\item[-]  Compute $[\mathbf{{\omega}_{j}^{k+1}}]_{pk}=\mathbf{Sub}([\mathbf{{\omega}_{j}^{k}}]_{pk}, [\mathbf{{\omega}_{j}^{k+1}}']_{pk})$ and broadcasts them to all users, where ${\omega}_{j}^{k+1}={\omega}_{j}^{k}-\frac{\eta}{\mathcal{B}\times N}\bigtriangledown {\omega}_{j}^{k}$.
  \end{itemize}
  \end{itemize}
\end{itemize}
\\
\hline
\end{tabular}
\caption{{Detailed description of  Hercules}}
\label{Detailed description of the scheme}
\vspace{-10pt}
\end{table*}
\renewcommand\tablename{TABLE}
\renewcommand \thetable{\Roman{table}}
\setcounter{figure}{4}
\setcounter{table}{6}

\section{Datasets Used in Experiments}
\label{AP:DATA}
Consistent with POSEIDON \cite{sav2020poseidon}, we choose the following  public datasets in our experiments.
\begin{packeditemize}
\item [a.]The Breast Cancer Wisconsin dataset (BCW) \cite{lavanya2011analysis}, which has a total of $|D|=699$ samples, and the dimension of each sample is $h=9$, with  $\omega_{\mathbb{L}}=2$, where $\omega_{\mathbb{L}}$ represents the number of labels (also the number of neurons in the last layer of the neural network (NN)).

\item[b.]MNIST dataset \cite{deng2012mnist} with $|D|=70,000$, $h=28\times 28$, and $\omega_{\mathbb{L}}=10$.

\item [c.]The Epileptic seizure recognition (ESR) dataset \cite{ESR} with $|D|=11,500$, $h=179$, and $\omega_{\mathbb{L}}=2$.

\item [d.]The default of credit card clients (CREDIT) \cite{CREDIT} with $|D|=30,000$, $h=23$, and $\omega_{\mathbb{L}}=2$.
\item [e.]The street view house numbers (SVHN) dataset \cite{SVHN} with colored images (3 channels),  with $|D|=600,000$, $h=3\times 32\times 32$, and $\omega_{\mathbb{L}}=10$.
\item [f.]CIFAR-10 and CIFAR-100 with  $|D|=60,500$,  $h=3\times 32\times 32$, $\omega_{\mathbb{L}}=10$, and $\omega_{\mathbb{L}}=100$, respectively.
\end{packeditemize}
We convert SVHN to gray-scale to reduce the number of channels. In addition, we make the dimension of each weight matrix to the nearest power of 2 by padding zeros. As a result, we actually train the NN on the CREDIT, ESR, and MNIST datasets with the feature dimensions of 32, 256, 1024, respectively. Since SVHN is grayed out, the dimension of its features already meets the power of 2. On the other hand,  we generate synthetic data to evaluate the scalability of the system, where we record the system performance under different features or the number of samples. For simplicity,  we uniformly and randomly distribute all the above datasets to users. Note that the distribution of data among different users may affect the accuracy of the model, but this is orthogonal to the focus of this paper. Since  \Name is designed for  general federated learning scenario, all existing processing strategies for data distribution can be seamlessly integrated into  \Name.

\begin{table*}
\footnotesize
\centering
\caption{ {Comparison with existing MPC-based works}}
\label{Comparison with existing MPC-based works}
\begin{threeparttable}
\begin{tabular}{|c|c|c|c|c|c|c|c|c|c|c|}
\Xhline{1pt}
\multirow{2}*{\diagbox{Index} {Schemes}}&ABY2.0&Trident&SWIFT&CRYPTFLOW&XONN&BLAZE&ABY$^3$ &SecureML&GALA&\multirow{2}*{\Name}\\
&\cite{patra2021aby2}&\cite{Trident}&\cite{koti2021swift}&\cite{kumar2020cryptflow}&\cite{riazi2019xonn}&\cite{patrablaze}& \cite{mohassel2018aby3}&\cite{mohassel2017secureml}&\cite{zhang2021gala}&\\ \hline
MPC Setup&2PC&4PC&3/4PC&3PC&2PC&3PC&3PC&2PC&2PC& N-party\\\hline
Inference&\Checkmark&\Checkmark&\Checkmark&\Checkmark&\Checkmark&\Checkmark&\Checkmark&\Checkmark&\Checkmark&\Checkmark*\\\hline
Training&\Checkmark&\Checkmark&\Checkmark&\XSolidBrush&\XSolidBrush&\XSolidBrush&\Checkmark&\Checkmark&\XSolidBrush&\Checkmark\\\hline
Adversarial model &1 P&1 A/P&2 A/P&2 P&1P& 1A&1 A/P&1P&1P&N-1P\\\hline
Collusion & \XSolidBrush &\XSolidBrush&2&\XSolidBrush&\XSolidBrush&\XSolidBrush&\XSolidBrush&\XSolidBrush&\XSolidBrush&N-1\\\hline
Techniques & GC, SS &GC, SS&SS&GC, SS&GC, SS&GC, SS&GC, SS&GC, SS&GC, HE&HE\\\hline
 Linear operation& \Checkmark&\Checkmark &\Checkmark&\Checkmark&\Checkmark&\Checkmark&\Checkmark&\Checkmark&\Checkmark&\Checkmark\\\hline
 Conv. operation& \Checkmark& \Checkmark&\Checkmark&\Checkmark&\Checkmark&\XSolidBrush&\Checkmark&\Checkmark&\Checkmark&\Checkmark\\\hline
 Pooling operation& \Checkmark& \Checkmark&\Checkmark&\Checkmark&\Checkmark&\XSolidBrush&\Checkmark&\Checkmark&\Checkmark&\Checkmark\\\hline
\end{tabular}
\begin{tablenotes}
        \item[] \footnotesize{2PC  represents the secure multi-party computing protocol between the two parties, and the rest can be deduced by analogy. A/P means active or passive adversary. \Checkmark* indicates that \Name also supports N-party distributed inference, because inference can be seen as a sub-process of training.}
      \end{tablenotes}
      \end{threeparttable}
\end{table*}
\section{Neural Network Structure}
\label{AP:NNS}
We use the following NN structures to train specific datasets.
\begin{packeditemize}
\item [a.]We train a 2-layer fully connected NN with BCW, ESR and CREDIT data sets, where each layer contains $64$ neurons. Similarly, the same structure was used to test the cost of \Name on the synthetic dataset.

\item[b.] We train a three-layer fully connected NN with 64 neurons in each layer for the MNIST and SVHN datasets.

\item [c.] We train two models for CIFAR-10 dataset. Specifically, (i) a CNN structure contains $2$ convolutional layers (CV), an average pooling layer, a max pooling layer with kernel size of $2\times 2$, and two fully connected layer (FC) which contains $128$ and $10$ neurons, respectively. We called this structure as CIFAR-10-N1. (ii) a CNN structure contains $4$ CV with kernel size of $3\times 4$,  an average pooling layer and a max pooling layer with kernel size of $2\times 2$, and 2 FC which contain $128$ and $10$ neurons, respectively. We called this structure as CIFAR-10-N2.

\item [d.]For CIFAR-100, we train a  CNN structure consisting of  $6$ CV with kernel size of $3\times 4$,  an average pooling layer and a max pooling layer with kernel size of $2\times 2$, and 2 FC  which contain $128$  neurons for each layer.
\end{packeditemize}
We changed the number of filters from 3 to 16 for all CV layers. The number of global iterations used to train the above models is 100; 500; 600; 1000; 18,000; 16,800; 25,000 and 54,000 for BCW, CREDIT, ESR, MNIST, SVHN,   CIFAR-10-N2, CIFAR-10-N1, and CIFAR-100, respectively. The local batch size of each user is set to 10, thus the global batch size is 100 for 10 users, and 500 for 50 users.
\section{Learning Extensions}
\label{AP:Learning Extensions}
\subsection{Asynchronous distributed learning with uneven data distribution}
In \Name, we rely on all users online to complete the global update of the gradient under the ciphertext. However, \Name also supports asynchronous neural network learning, which can be realized by the server only receiving data from users within a certain threshold time. In this case, the training of the neural network mainly benefits from the local data of users with good network status.  Noted that the smooth execution of the distributed bootstrapping operation  $\mathbf{DBootstrap}$ requires the participation of all users. This operation instead of being done in a centralized manner usually requires more computation and communication overhead.

In  \Name, we evenly distribute the data to each user for the simplicity of the experiment. In fact, the uneven distribution of users' data and asynchronous gradient descent will inevitably affect the accuracy of the model. How to deal with these problems has been extensively studied \cite{tang2018d,tang2019doublesqueeze},  which are orthogonal to this work. Note that our privacy protection method is independent of the distribution of user data and the way of gradient update. Hence, all existing works can be seamlessly integrated into \Name.

\subsection{Training on Other Neural Networks}

 In this work, we focus on deploying \Name on MLPs and CNNs, and demonstrate the performance of our packing scheme and function approximation method under these NNs. For other structures, such as long short-term memory (LSTM)\cite{baytas2017patient}, recurrent neural network (RNN)\cite{li2018independently} and residual neural network (ResNet) \cite{fablet2017bilinear}, \Name needs to modify the local gradient update process according to their forward and backward pass operations. For example, ResNet has skip connections to skip certain layers, so after skipping a layer, the shape of the encrypted ciphertext should be aligned according to the weight matrix to be multiplied. This can be ensured by using the rotation function of CKKS (rearrange the slots of the ciphertext). We believe that these modifications are not significant. Although POSEIDON is tailored for MLPs and CNNs, in essence, \Name can be used to quickly calculate any mathematical operation under ciphertext, which constitutes the main body of other neural network structures.

\setlength{\abovecaptionskip}{10pt}
\setlength{\belowcaptionskip}{5pt}
\section{Microbenchmarks of \Name and POSEIDON}
\label{Microbenchmarks}
\begin{table}
\footnotesize
\caption{Microbenchmarks of \Name and POSEIDON}
\label{Microbench}
\begin{tabular}{l|l|l}
\hline
Functionality& Execution time (s)& Comm. (MB)\\
\midrule
Approx\_Sigmoid & 0.017 | \textbf{0.017} & --  \\
Appro\_ReLu & 0.01 | \textbf{0.09} & -- \\
Appro\_Softmax & 0.07 | \textbf{0.07} &  -- \\
Average-pooling &0.034|\textbf{0.034} & -- \\
Max-pooling &5.73|\textbf{1.91} &23.5|\textbf{6.9}\\
DBootstrap &  0.09 | \textbf{0.09} & 4.5 | \textbf{1.5} \\
FC layer & 0.097 | \textbf{0.0015} &  -- \\
FC layer-backprop & 0.14 | \textbf{0.002} & -- \\
CV layer & 0.04 | \textbf{0.0006}  & -- \\
CV layer-backprop & 0.06 | \textbf{0.0009}& -- \\
DKeySwitch & 0.06 | \textbf{0.06} &  22.14 | \textbf{22.14}\\\hline
\end{tabular}
\end{table}
We present the microbenchmark  cost of \Name and POSEIDON in performing different functions. As shown in Table~\ref{Microbench}, the bold font  represents the overhead of  \Name.  For  experimental configurations, We consider the number of users $N=10$, the dimension of each sample is $h=32$, the number of neurons in each layer is 64 or with kernel size  $3\times 3$, and the  dimension of  cyclotomic polynomial ring  used in CKKS is set $\mathcal{N}=2^{13}$. All the costs represent the processing of 1 sample per party.

We record the computation and communication costs (in an amortized way) of calculating FC, CV, FC backpropagation, CV backpropagation, and different activation functions in \Name and POSEIDON. For the activation functions sigmoid and sofmax, \Name and POSEIDON exhibit the same overhead because they both use the same function approximation method. For \texttt{ReLU}, since \Name uses a 20 bit-precision  composite polynomial to fit the original function, while POSEIDON simply uses a polynomial with a degree of 3 to approximate the original \texttt{ReLU} (using the least squares method),  it makes the cost of  \Name larger than POSEIDON but derives a higher approximation accuracy. Note that given the same error bound, the time required to calculate the polynomial generated by the least square method under ciphertext is much longer than ours.  For max pooling, POSEIDON replaces it with Chebyshev function interpolation, thereby obtaining a polynomial of degree 31 with a precision of 7 bits. Experimental results show that our 7-bit precision composite polynomial is much less expensive than POSEIDON. This is mainly due to the characteristics of composite polynomials. As discussed before, a composite function with degree of $deg(G)$ can be calculated with a complexity of $O(\log(deg (G)))$,  while the computation complexity of calculating any function $G$ is at least $\Theta(\sqrt{deg (G)})$.

The operations in FC, CV and their backpropagation are mainly composed of matrix multiplication. As shown in Table~\ref{Microbench}, \Name is far superior to POSEIDON in terms of calculating homomorphic matrix multiplication. Specifically, POSEIDON adopts AP to achieve fast SIMD calculations.  For the multiplication of two $h\times h$-dimensional matrices, the complexity of the homomorphic rotation operation required by AP is $\max_{i\in [\mathbb{L}]}(\omega_i\times\log(h\times \omega_i))$.  For \Name, as shown in Table~\ref{Complexity of Algorithm 3}, the complexity required for the matrix multiplication is only $3h+5\sqrt{h}$. Moreover, AP requires multiple copies and zero padding operations for each row or column of the input matrix, (depending on the number of neurons in each hidden layer, and the absolute value of the difference between the row or column dimension of the matrix and the corresponding hidden layer neuron. On the contrary, our method does not require additional element copying except for a small amount of zero padding in the initial stage to facilitate calculations. Therefore, \Name obviously exhibits lower communication overhead.

\section{Storage Overhead}
\label{Comparison of storage cost}
\begin{table}
\centering
\small
\caption{Comparison of storage cost}
\label{Storage overhead}
\begin{tabular}{|c|c|c|}
\Xhline{1pt}
\multirow{2}*{Datasets}& \multicolumn{2}{c|}{Storage cost (GB) }\\ \cline{2-3}
&\Name &POSEIDON\\ \hline
BCW& $0.0012$  & 0.095\\ \hline
ESR& $0.184$ & $11.778$ \\ \hline
CREDIT& $0.061$ & $3.948$\\ \hline
MNIST& $4.906$ & $314.025$\\ \hline
SVHN& $54.9$ & $3515.6$\\ \hline
CIFAR-10& $16.47$ & $2109.37$\\ \hline
CIFAR-100& $16.47$ & $2109.37$\\ \Xhline{1pt}
\end{tabular}
\end{table}
As shown in TABLE~\ref{Storage overhead}, we also count the storage overhead required for training each scheme on different datasets,  where the storage overhead is mainly dominated by the size of the ciphertext that each user needs to save locally. We observe that the storage overhead required by \Name on each dataset is much lower than that of POSEIDON. Since POSEIDON designs a  alternating packing (AP)  method to  pack ciphertext in the encryption process,  this leads to multiple copies and zero padding operations for each ciphertext, depending on the number of neurons in
each hidden layer.  On the contrary, our method does not require additional element copying except for a small amount of zero padding in the initial stage to facilitate calculations.

\section{Comparison with Other Prior Works}
 Beyond the HE-based schemes, a wealth of MPC solutions based on interaction between multiple servers have been proposed to design privacy-preserving neural network training and prediction frameworks. These solutions may be not practical in the real-world setting without multiple servers. Nevertheless, we still make a rough comparison between these works and \Name to demonstrate the merits and demerits of each scheme. As shown in Table~\ref{Comparison with existing MPC-based works} in the Appendix, MPC protocol with multiple servers relies on splitting the training task into two or more servers, where the servers are usually assumed to be non-colluding.  Then, the state-of-the-art secret sharing methods, including arithmetic sharing \cite{patra2021aby2}, boolean sharing \cite{mohassel2018aby3}, and Yao's garbled circuit \cite{kumar2020cryptflow} are carefully integrated to efficiently implement various mathematical operations under the ciphertext. It is computationally cost-effective, avoiding the high communication overhead among large-scale users. For example,  given the MNIST dataset, we roughly compare the computational cost of  \Name, ABY$^3$ \cite{mohassel2018aby3}, and XONN  \cite{riazi2019xonn} in training the same network architecture. Here we use a three-layer architecture with 128 neurons in each layer. The number of users in \Name is set to 3 and the global training epoch is set to 15. Experimental results show that \Name requires a total of 24.3 hours to train such a model, while  ABY$^3$  and XONN need 1.02 and 0.58 hours, respectively.  We remind that \Name is operated in a different  scenario (FL) and threat model from those MPC-based schemes. It supports more participants with collusion, and the cost only grows linearly with the number of users.  In contrast, those MPC-based solutions require to outsource training tasks among limited computing servers, which may be impractical in some scenarios.

\end{document}